\documentclass[copyright,creativecommons]{eptcs}
\providecommand{\event}{LSFA 2026}

\usepackage{iftex}

\ifpdf
  \usepackage{underscore}         
  \usepackage[T1]{fontenc}        
\else
  \usepackage{breakurl}           
\fi

\usepackage{tikz}
\usepackage{amssymb}
\usepackage{amsthm}
\usepackage{amsmath}
\usepackage[capitalize]{cleveref}
\usepackage{caption}
\usepackage{subcaption}
\usepackage{newunicodechar} 
\usepackage{inconsolata}
\usepackage{fancyvrb}
\usepackage{bussproofs}
\usepackage{enumitem}

\usetikzlibrary{positioning}

\newunicodechar{∀}{\ensuremath{\mathnormal\forall}}
\newunicodechar{Γ}{\ensuremath{\mathsf\Gamma}}
\newunicodechar{Δ}{\ensuremath{\mathsf\Delta}}
\newunicodechar{∪}{\ensuremath{\cup}}
\newunicodechar{∃}{\ensuremath{\exists}}
\newunicodechar{Π}{\ensuremath{\mathsf\Pi}}
\newunicodechar{λ}{\ensuremath{\mathsf\lambda}}
\newunicodechar{⊢}{\ensuremath{\vdash}}
\newunicodechar{∶}{{\ensuremath{:}}}
\newunicodechar{β}{{\ensuremath{\beta}}}
\newunicodechar{α}{{\ensuremath{\alpha}}}
\newunicodechar{δ}{{\ensuremath{\delta}}}
\newunicodechar{ε}{{\ensuremath{\varepsilon}}}
\newunicodechar{ρ}{{\ensuremath{\rho}}}
\newunicodechar{Σ}{{\ensuremath{\Sigma}}}
\newunicodechar{≃}{{\ensuremath{\simeq}}}
\newunicodechar{₌}{{\ensuremath{{}_=}}}
\newunicodechar{₀}{{\ensuremath{{}_0}}}
\newunicodechar{₁}{{\ensuremath{{}_1}}}
\newunicodechar{₂}{{\ensuremath{{}_2}}}
\newunicodechar{₃}{{\ensuremath{{}_3}}}
\newunicodechar{₄}{{\ensuremath{{}_4}}}
\newunicodechar{ₛ}{{\ensuremath{{}_\mathsf{s}}}}
\newunicodechar{⁻}{\ensuremath{{}^{-}}}
\newunicodechar{∈}{{\ensuremath{\in}}}
\newunicodechar{∉}{{\ensuremath{\not\in}}}
\newunicodechar{≡}{{\ensuremath{\equiv}}}
\newunicodechar{σ}{{\ensuremath{\sigma}}}
\newunicodechar{∙}{{\ensuremath{\bullet}}}
\newunicodechar{‚}{{\ensuremath{,}}}
\newunicodechar{Λ}{{\ensuremath{\mathsf\Lambda}}}
\newunicodechar{𝒱}{\ensuremath{\mathcal{V}}}
\newunicodechar{𝒞}{\ensuremath{\mathcal{C}}}
\newunicodechar{𝒜}{\ensuremath{\mathcal{A}}}
\newunicodechar{ℛ}{\ensuremath{\mathcal{R}}}
\newunicodechar{ℕ}{\ensuremath{\mathbb{N}}}
\newunicodechar{⊥}{\ensuremath{\bot}}
\newunicodechar{ι}{\ensuremath{\iota}}
\newunicodechar{≟}{\ensuremath{\stackrel{?}{=}}}
\newunicodechar{⇂}{\ensuremath{\downharpoonright}}
\newunicodechar{∼}{\ensuremath{{\sim}}}
\newunicodechar{▹}{\ensuremath{{\triangleright}}}
\newunicodechar{𝒮}{\ensuremath{\mathcal{S}}}
\newunicodechar{↔}{\ensuremath{\leftrightarrow}}
\newunicodechar{⇉}{\ensuremath{\rightrightarrows}}
\newunicodechar{⇒}{\ensuremath{\Rightarrow}}
\newunicodechar{⇔}{\ensuremath{\Leftrightarrow}}
\newunicodechar{⊔}{\ensuremath{\sqcup}}
\newunicodechar{◅}{\ensuremath{\triangleleft}}
\newunicodechar{;}{\ensuremath{;}}
\newunicodechar{⊆}{\ensuremath{\subseteq}}
\newunicodechar{⊎}{\ensuremath{\uplus}}
\newunicodechar{≺}{\ensuremath{\prec}}
\newunicodechar{≈}{\ensuremath{\approx}}
\newunicodechar{τ}{\ensuremath{\tau}}
\newunicodechar{⋯}{\ensuremath{\dots}}
\newunicodechar{ℓ}{\ensuremath{\ell}}
\newunicodechar{∷}{\ensuremath{\Colon}}
\newunicodechar{↠}{\ensuremath{{\twoheadrightarrow}}}
\newunicodechar{⋆}{\ensuremath{{}^*}}
\newunicodechar{⊤}{\ensuremath{\top}}

\newcommand{\websiteurl}{https://surciuoli.github.io/pts-consistency-impred}
\newcommand{\sourcelink}[2]{\href{\websiteurl/links/#1}{#2}}

\newcommand{\redex}[2]{\ensuremath{#1 \triangleright_\beta #2}}

\newcommand{\ok}[1]{\ensuremath{#1\,\mathsf{ok}_{\mathsf{s}}}}
\newcommand{\types}[3]{\ensuremath{#1 \vdash_{\mathsf{s}} #2 : #3}}
\newcommand{\dom}[1]{\ensuremath{\mathsf{dom}\,#1}}
\newcommand{\Starg}{\ensuremath{\mathsf{Star}}}
\newcommand{\Star}[3]{\ensuremath{\mathsf{Star}\,#1\,#2\,#3}}

\newcommand{\step}[2]{\ensuremath{#1 \rightarrow_\beta #2}}

\newcommand{\alphaconv}[2]{\ensuremath{#1\ {\sim}_\alpha\ #2}}
\newcommand{\conv}[2]{\ensuremath{#1 \simeq_\beta #2}}

\newcommand{\convgspc}{\ensuremath{\,{\simeq}_\beta\,}}
\newcommand{\manystep}[2]{\ensuremath{#1\ {\rightarrow}{\beta}{*}_0\ #2}}
\newcommand{\manystepg}{\ensuremath{\_{\rightarrow}{\beta}{*}_0\_}}
\newcommand{\binrel}{\ensuremath{\mathcal{S}}}

\newcommand{\fv}[1]{\ensuremath{\mathsf{fv}\,#1}}

\renewcommand{\parallel}[2]{\ensuremath{#1 \rightrightarrows #2}}
\newcommand{\parallelg}{\ensuremath{\_{\rightrightarrows}\_}}
\newcommand{\parallelbeta}{\ensuremath{{\rightrightarrows}_\beta}}
\newcommand{\vdashstd}{\ensuremath{\vdash_{\mathsf{s}}}}
\newcommand{\typerule}[1]{\ensuremath{{\vdash}}\texttt{#1}}
\newcommand{\nf}[1]{\ensuremath{\mathsf{nf}\,#1}}
\newcommand{\neu}[1]{\ensuremath{\mathsf{ne}\,#1}}
\newcommand{\nfg}{\ensuremath{\mathsf{nf}}}
\newcommand{\neug}{\ensuremath{\mathsf{ne}}}
\newcommand{\Nf}[1]{\ensuremath{\mathsf{Nf}\,#1}}
\newcommand{\Neu}[2]{\ensuremath{\mathsf{Ne}_{#1}\,#2}}
\newcommand{\Nfg}{\ensuremath{\mathsf{Nf}}}
\newcommand{\Neg}{\ensuremath{\mathsf{Ne}}}
\newcommand{\IsVarg}{\ensuremath{\mathsf{IsVar}}}

\newcommand{\app}{\ensuremath{\cdot}}
\newcommand{\quant}{\ensuremath{\reflectbox{Q}}}
\newcommand{\append}{\ensuremath{\,{+}\!\!{+}\,}}

\newcommand{\RelAux}[2]{\ensuremath{#1 \rightarrow_\binrel #2}}
\newcommand{\StarAux}[2]{\ensuremath{#1 \twoheadrightarrow_\binrel #2}}

\fvset{
  fontsize=\small,
  numbers=none
}

\newtheoremstyle{theo}
  {}
  {}
  {}
  {}
  {\bfseries}
  {.}
  { }
  {}

\theoremstyle{theo}  
\newtheorem{theorem}{Theorem}
\newtheorem{lemma}{Lemma}

\newcommand{\myitem}[1]{%
  \refstepcounter{enumi}%
  \item[\sourcelink{#1}{(\roman{enumi}})]
}
\Crefname{enumi}{Lemma}{Lemmas}

\newcommand{\thmitem}[1]{%
  \refstepcounter{enumi}%
  \item[\sourcelink{#1}{(\roman{enumi}})]
}
\Crefname{thmitem}{Theorem}{Theorem}

\newtheoremstyle{linked}
  {}
  {}
  {}
  {}
  {\bfseries}
  {.}
  { }
  {\sourcelink{\thislink}{\thmname{#1} \thmnumber{#2}}\thmnote{ (#3)}}

\theoremstyle{linked}
\newtheorem{innlinklemma}[lemma]{Lemma} 
\newenvironment{linklemma}[1]
 {\def\thislink{#1}\innlinklemma}
 {\endinnlinklemma}
\Crefname{innlinklemma}{Lemma}{Lemmas}

\Crefname{innlinkcoro}{Corollary}{Corolaries}

\title{A Machine-checked Proof of Consistency for Impredicative Pure Type Systems}
\author{Sebasti\'an Urciuoli \institute{Universidad ORT Uruguay}\email{urciuoli@ort.edu.uy}}

\def\authorrunning{S.~Urciuoli}
\def\titlerunning{Consistency of Impredicative PTS}

\begin{document}

\maketitle

\begin{abstract}
In this paper we continue assessing the feasibility of the approach to the mechanization of type theory by using classical syntax and Stoughton's multiple substitutions and report some substantial progress.
We present formal proofs of confluence for beta-reduction and by using Takahashi's revision of Tait and Martin-L\"of's proof, subject reduction for the entire family of the Pure Type Systems and consistency for some impredicative subclass, assuming normalization.
As to the proof of confluence, we also develop a theory of alpha-commutative relations which, in our view, entails a clearer presentation and treatment of the problem than in similar developments.
Finally, we assess general merits and drawbacks of the approach.
The whole development has been machine-checked using Agda.
\end{abstract}

\section{Introduction}

The first ever formal definition of the substitution operation for the (pure) $\lambda$-calculus can be probably attributed to Curry-Feys \cite{curry1958}. Such definition is not primitive recursive; in the complex case of $\lambda$-abstractions it invokes itself twice, once to rename the bound variable in order to prevent name capture, and another to perform the actual substitution. 
The second time is on an expression that is not a strict component of the original.
Formally:
\[
(\lambda x M)[y:=N] = \lambda z(M[x:=z][y:=N])
\]
where $z$ is some fresh name for $N$. 
An unfortunate consequence of the nature of this definition is that
proofs about 
meta-theoretical results 
that mention substitution, 
either directly or indirectly (virtually all), have to be carried out by well-founded induction on the length of the $\lambda$-terms in order to have fit hypotheses. 
It does not seem feasible to undertake any major mechanization considering this definition. 

There have been a lot of alternatives proposed to avoid renaming the bound variables during substitution, all of them with their particular benefits and drawbacks. 
Possibly the most popular one is that suggested by de Bruijn and which removes the names from the syntax altogether (\textit{de Bruijn indices}, dBI) \cite{debruijn1972}.
Instead, variables are identified with natural numbers; those that occur bound to some abstraction denote the count of binders in between one must traverse in the abstract syntax tree (AST) bottom-up until finding its binding location, while the free or real ones refer to its position in the context of declarations (assuming the term we are referring to is the subject of some judgment).
Although this syntax may be convenient for machines, it is not so for humans, hence to fully address type theory under this paradigm we believe one should also relate the AST with its interface, i.e., the concrete syntax, and which is rarely done.\footnote{Some mechanizations do address this issue, e.g., \cite{barras96b}.}
A more conservative approach is to use a mix of names for the free variables and natural numbers for the bound ones. This technique is known as locally nameless \cite{chargueraud2012}.
Actually, it may be regarded as a specialization of the more general technique in which any two different sorts of names are used 
(McKinna and Pollack use this technique to formalize several results about an extension of the PTS, the Cumulative Type System or CTS, in LEGO \cite{mckinna93,mckinna99,pollack94phd}).
Under this approach 
there are some expressions that do not have an ordinary interpretation, i.e., those that mention bound names which are actually not bound to any abstraction. 
Since substitution does not perform any renaming, it cannot operate on these, otherwise variable capture might occur. 
To avoid these expressions, 
a wellformedness predicate 
becomes
necessary and which ends up polluting most of meta-theoretical results.
Yet another well-known approach
involves the use of functions from the meta-level to represent binders, i.e., higher-order abstract syntax (HOAS) \cite{pfenning1988}. This coding technique delegates the problems arising from renaming variables to the meta-level. 
To the best of our knowledge, 
the only mechanization of type theory that uses this technique (a restricted form actually, \textit{weak} HOAS \cite{despeyroux1995}) ``comparable'' to our presentation (in the sense that it formalizes complex results such as cut, syntactic validity, product injectivity, subject reduction, etc., for a type theory with dependent types) is that by Urban et al. \cite{urban2011} and by using Isabelle/HOL \cite{nipkow2002}.\footnote{There are a number of other interestings formalizations of relevant meta-theoretical results for type theory using HOAS, e.g., normalization both for STLC \cite[Section~4.1]{abel2019} and System~F \cite{donelly2007}, just to mention some. Nevertheless, they target on quite different systems and/or properties to the ones we focus on in here, it bears repeating.}
Although it is a remarkable work and much more extensive than ours, it worths mentioning a couple of spiny limitations it has: it is not possible to generate executable code from the definitions and decidability questions such as typeability are not possible to formulate, at least directly.

There is, however, a quite elegant and simple solution that retains the classical syntax 
and which deserves attention as pointed out in \cite{tasistro2015}. Stoughton proposed to use simultaneous substitutions in order to perform the renaming of the bound variables at the same time the original substitution takes place \cite{stoughton1988}. 
If we denote a simultaneous substitution by a function $\sigma$ from the variables to the $\lambda$-terms, and the operation itself by $\_\bullet\_$ then the equation for $\lambda$-abstractions becomes:
\[
(\lambda x M)\bullet\sigma = \lambda y (M \bullet \sigma,y:=z)
\]
where $z$ is some fresh name and $\_,\_{:=}\_$ is the obvious update operator on substitutions that overrides some given image.
This operation can actually be defined by structural recursion since it only requires one self-invocation on a strict sub-expression.

This approach, which materialized itself as an Agda library in \cite{tasistro2015}, has been used since then to formalize some results about different type theories: 
in \cite{copello2017} it was used to formalize, first, the Church-Rosser (CR) theorems following the classical proof by Tait and Martin-L\"of for the pure $\lambda$-calculus \cite{barendregt84}, and secondly, subject reduction (SR) for the Simply-typed $\lambda$-calculus (STLC) à la Curry \cite{hindley1997}; 
in \cite{copes2018} the standardization theorem was mechanized following a proof by Kashima \cite{kashima2000}; 
in \cite{urciuoli2023} the framework was extended with constants in the syntax and then it was used to formalize the proof of strong normalization for System~T by Girard \cite{girard1989}, and; 
finally, in \cite{urciuoli2025} the framework was extended once again, this time in several directions, and used to formalize some dependently-typed $\lambda$-calculus, namely the Pure Type Systems (PTS) \cite{barendregt1991,barendregt92}, along with some meta-theoretical properties (thinning or weakening, cut, syntactic validity and closure under $\alpha$-conversion).\footnote{Some of the main changes introduced in the second revision of the framework were: the type of variables was generalized from natural numbers to  any denumerable type; the syntax was updated with Church-style $\lambda$-abstractions and $\Pi$-types were added; a key definition, namely \textit{restrictions}, which denotes the confinement of substitutions to finite domains, was updated by a more flexible one in order to establish some results more precisely in the new context.}
The PTS is a framework for the study of several related type theories with $\lambda$-terms. 
It can be seen as a generalization of the $\lambda$-cube, so as such, 
it contains many interesting systems: 
STLC; System~F \cite{girard1989}; the Edinburgh Logical Framework (LF) \cite{barendregt92,harper93,harper2005,sorensen2006}, 
the Calculus of Constructions (CC) \cite{coquand1988}; etc.
Ultimately, we would like to formalize a proof of decidability for some subset of the PTS; from this result we would obtain correct-by-constructions type-checking algorithms. 

The aforementioned work revealed that actually only a couple of lemmas out from dozens had to be proven by well-founded induction.
Also, the numbers of lines of code (LoC) showed that the size of the development did not explode by any means and it stayed more or less on par with similar mechanizations using dBI or some of its variants. All in all, the approach seems to be feasible, however, more work should be carried on to have a better judgment, specially when using dependent types.

\paragraph{Contributions}
We shall continue with our approach to type theory and use Agda to formalize: first, the CR theorems for the underlying syntax of the PTS and by using Takahashi's revision of Tait and Martin-L\"of's proof \cite{takahashi1995} (similarly to McKinna and Pollack but from a fresh and cleaner perspective, i.e., considering $\alpha$-conversion explicitly, aside from the obvious fact that we use a different approach to the syntax); 
second, SR (following some ideas from the formalizations by McKinna and Pollack we well), and; 
thirdly, consistency in the empty context for some impredicative class of PTS and by extending Coquand's pen-and-paper proof for CC \cite{coquand1990}.\footnote{Consistency of CC has been completely formalized using dBI \cite{coqincoqSources}, though it has not been published to the best of our knowledge. 
There is also a partial mechanization of consistency (assuming normalization) for a similar theory to Coq in Coq using dBI and an inductively-defined empty type in the context \cite{sozeau2025}.}
Normalization is assumed.\footnote{It does not seem possible to prove normalization of impredicative theories such as CC in Agda (see \cref{sec:conclusions}).}

\paragraph{Sources}
This work has been fully verified using Agda v2.6.2.2 \cite{agda2622} and the standard library v1.7.1 \cite{agdastdlib171}. 
The source can be downloaded from \cite{sources} and browsed in \cite{html}.
If this text is read on a computer, sources of definitions and lemmas can be examined by clicking on their respective links, e.g., \sourcelink{Section2/Syntax.html}{syntax}.

\paragraph{Outline}
The structure of this paper is as follows. 
In the next section we shall present the framework of Stoughton's multiples substitutions for the syntax of the PTS from \cite{urciuoli2025} so this work is self-contained. 
In \cref{sec:confluence} we present our proof of confluence.
In \cref{sec:sr} we introduce the PTS and its main meta-theoretical properties also from \cite{urciuoli2025}, and we formalize SR as well.
In \cref{sec:consistency} we formalize consistency.
Lastly, in \cref{sec:conclusions} we compare our work with other related work and draw some conclusions. 

\section{The Framework of Multiple Substitutions}
\label{sec:stoughton}

This section presents the Agda library developed in \cite{urciuoli2025} for the meta-theory of the underlying syntax of the PTS using Stoughton's substitution. 
Among many other definitions, it contains the operation of substitution, $\alpha$- and $\beta$-conversion relations, as well as some compatibility results between these. 

\subsection{Syntax}

Let $\mathcal{C}$, the constants, be any type and ranged with the letter $c$, and let $\mathcal{V}$, the variables, be any \sourcelink{Section2/Enum.html}{enumerable} type and ranged with letters $x$, $y$, etc. 
The abstract \sourcelink{Section2/Syntax.html}{syntax} of the $\lambda$-terms is given by the grammar below:
\[
(\Lambda) \quad M, N, A, B ::=  \mathsf{c}\ s\ |\ \mathsf{v}\ x\ |\ \lambda[x:A]M\ |\ \Pi[x:A]M\ |\ M \cdot N
\]
 
We define by recursion on the syntax the function \Verb|fv : Λ → List 𝒱| that returns the list of \sourcelink{Section2/FreeVariables.html}{free variables} in a given $\lambda$-term in the usual way; we omit its definition.
By using \Verb|fv| we define the predicates of \sourcelink{Section2/Occurrence.html}{occurrence} of a variable in a $\lambda$-term, with type \Verb|_*_ : 𝒱 → Λ → Set|, and its opposite, \sourcelink{Section2/Freshness.html}{freshness}, \Verb|_#_ : 𝒱 → Λ → Set|.
We refer the interested reader to the sources.

\subsection{Substitutions}

In our framework, \sourcelink{Section2/Substitutions.html}{substitutions} are multiple or simultaneous and they are identified with functions from variables to $\lambda$-terms:
\begin{Verbatim}
Sub = 𝒱 → Λ
\end{Verbatim}
We shall use letters $\rho$, $\sigma$ and $\tau$, possibly with primes, to range them.

We define \Verb|ι : Sub| as the \sourcelink{Section2/Identity.html}{identity} substitution, i.e., the function that maps every variable to itself as a $\lambda$-term.
We also define an \sourcelink{Section2/Update.html}{update} operator on substitutions \Verb|_‚_:=_ : Sub → 𝒱 → Λ → Sub| that overrides the image of some element in the domain, i.e., $(\sigma , x := M)y$ yields $M$ if $x$ equals to $y$ and $\sigma y$ otherwise. 

To define relations and prove properties about substitutions by extension it turns out convenient to confine their domains. Thus we introduce \sourcelink{Section2/Restrictions.html}{restrictions}, which are simply pairs of substitutions and list of names, to be written $(\sigma,xs)$:
\begin{Verbatim}
Res = Sub × List 𝒱
\end{Verbatim}

Now, to define Stoughton's substitution operation, first we have to be able to select some fresh name for the image of any given restriction; we will use these names to rename the bound variables in the $\lambda$-terms at issue in order to avoid name capture during the substitution operation.
Actually, this function is defined upon a more elementary one that choses a fresh name for a \textit{$\lambda$-term}. 
We will only show their types, and we we refer the reader to \cite{tasistro2015} for their definitions as well as for a description:
\begin{Verbatim}
X' : Λ → 𝒱
X : Res → 𝒱
\end{Verbatim}
Above, \Verb|X'| is to be read \sourcelink{Section2/ChiPrime.html}{\it chi prime}, and \Verb|X| as \sourcelink{Section2/Chi.html}{\it chi}. 
    
In the aforementioned work it was proven that the functions above actually returns a fresh name. Let \sourcelink{Section2/FreshnessRes.html}{freshness be extended to restrictions} by:
\begin{Verbatim}
x #⇂ (σ , xs) = ∀ y → y ∈ xs → x # σ y
\end{Verbatim}
Then we have:
\begin{lemma}
\begin{enumerate}[ref=\thelemma.(\roman{enumi}),itemsep=0pt]
	\myitem{Section2/LemmaChiPrime.html} 
	\label{lemma:chiPrime}
	\Verb|∀ xs → X' xs ∉ xs|
	
	\myitem{Section2/LemmaChi.html} 
	\label{lemma:chi}	
	\Verb|∀ σ xs → X (σ , xs) #⇂ (σ , xs)|
\end{enumerate}
\end{lemma}

Now the \sourcelink{Section2/Substitution.html}{substitution operation} can be defined by structural recursion:
\begin{Verbatim}
_∙_ : Λ → Sub → Λ
c k ∙ _ = c k
v x ∙ σ = σ x
M · N ∙ σ = (M ∙ σ) · (N ∙ σ)
λ[ x ∶ A ] M ∙ σ = λ[ y ∶ A ∙ σ ](M ∙ σ , x := v y) where y = X (σ , fv M - x)
Π[ x ∶ A ] B ∙ σ = Π[ y ∶ A ∙ σ ](B ∙ σ , x := v y) where y = X (σ , fv B - x)
\end{Verbatim}
\normalsize

The only interesting cases to mention in the definition of the substitution operation are that of $\lambda$-abstractions and $\Pi$-types. There, we always rename the bound variables to prevent any possible name clash. 
By using the same substitution to record such renaming we only need to invoke a recursive call once. 
The new name, $y$, is chosen so it is fresh for the image of the free variables in $M$ under $\sigma$ except for $x$ since it is renamed anyway. Actually, if we considered $x$ to build the list of free names in the image, then $y$ might not be the \textit{first} name available, since $x$ in $\sigma$ may introduce free names that are not troublesome but they would be still dismissed anyway.

We shall abbreviate \sourcelink{Section2/Unary.html}{unary substitutions} by: 
\begin{Verbatim} 
M [ x := N ] = M ∙ (ι ‚ x := N)
\end{Verbatim}

The following results were proven in \cite{urciuoli2025}:
\begin{lemma}
\normalfont
\small
\begin{enumerate}[ref=\thelemma.(\roman{enumi}),itemsep=0pt]
	\myitem{Section2/SubLemma1.html}
	\label{lemma:factorizeUpd}		
	\Verb|∀ {M N σ x} → M ∙ (σ ‚ x := (N ∙ σ)) ≡ M ∙ (ι ‚ x := N) ∙ σ|
	
	\myitem{Section2/SubLemma2.html}
	\label{lemma:expandUpd}
	\Verb|∀ {x z M N σ} → z ∉ fv M - x → M ∙ (σ ‚ x := N) ≡ M [ x := v z ] ∙ (σ ‚ z := N)|
\end{enumerate}
\end{lemma}

\subsection{Alpha-conversion}

\sourcelink{Section2/Alpha.html}{Alpha-conversion} is defined by induction using the next informal rules:

\begin{center}
\centering
\small
\AxiomC{}
\UnaryInfC{\alphaconv{k}{k}}
\bottomAlignProof
\DisplayProof
\quad
\AxiomC{}
\UnaryInfC{\alphaconv{x}{x}}
\bottomAlignProof
\DisplayProof
\quad
\AxiomC{\alphaconv{M}{M'}}
\AxiomC{\alphaconv{N}{N'}}
\BinaryInfC{\alphaconv{M \app N}{M' \app M'}}
\bottomAlignProof
\DisplayProof
\quad
\AxiomC{\alphaconv{A}{A'}}
\AxiomC{$M[x:=y] \equiv M'[x':=y]$}
\RightLabel{$\begin{cases}y \not\in\fv{M} - x\\y \not\in\fv{M'}-x'\end{cases}$}
\BinaryInfC{\alphaconv{\quant[x:A]M}{\quant[x':A']M'}}
\bottomAlignProof
\DisplayProof
\end{center}
where $\quant\in\{\lambda,\Pi\}$.

Again, the only cases worth some words are that of $\lambda$-abstractions and $\Pi$-types (\quant).
There, to derive that such syntactic classes are $\alpha$-convertible we require that their bodies be \textit{equal} after their respective bound variables have been replaced by a common fresh name. In \cite{urciuoli2025}, it was proven that this definition is equivalent to a more standard one that asks for their bodies to be \textit{$\alpha$-convertible} instead. This is because the substitution operation enforces a uniform renaming (by the choice function) that turns any two $\alpha$-convertible $\lambda$-terms syntactically equal.

The next results are from \cite{urciuoli2025}:
\begin{lemma}
\begin{enumerate}[ref=\thelemma.(\roman{enumi}),label=(\roman{enumi}),itemsep=0pt]
	\myitem{Section2/AlphaSub.html}
	\Verb|∀ {M M' N N' x} → M ∼α M' → N ∼α N' → M [ x := N ] ∼α M' [ x := N' ]|
	\label{lemma:compatAlphaSubUnary}
	
	\item $\alpha$-conversion is a \sourcelink{Section2/AlphaReflexive.html}{reflexive}, \sourcelink{Section2/AlphaSymmetric.html}{symmetric} and \sourcelink{Section2/AlphaTransitive.html}{transitive}.
	\label{lemma:alphaEquivalence}
\end{enumerate}
\end{lemma}

\subsection{Beta-conversion}
\label{sec:beta}

Next we will define small-step operational semantics: $\beta$-reduction.
We will diverge a bit from previous work in that we will treat \textit{many-step reduction} more rigorously and define it as a sequence of steps of $\beta$-contraction but not of $\alpha$-conversion;
in the next section we will develop a proof of the first CR theorem that fits best with this new characterization. 

Let \sourcelink{Section2/Beta.html}{$\beta$-contraction} be inductively 
defined with the single clause: 
\begin{prooftree}
\AxiomC{}
\LeftLabel{($\beta$)}
\UnaryInfC{\redex{(\lambda[x:A]M) \app N}{M[x:=N]}}
\end{prooftree}
Then \sourcelink{Section2/BetaReduction.html}{one-step $\beta$-reduction} (\Verb|_→β_|) is defined as its contextual closure;
\sourcelink{Section2/ManyStepBeta.html}{many-step $\beta$-reduction} (\Verb|_→β*₀_|) as the star closure of the former; and;
\sourcelink{Section2/BetaConversion.html}{$\beta$-conversion} (\Verb|_≃β_|) as the equivalence closure of \Verb|_→β_| augmented with $\alpha$-conversion as well (see the next paragraph for the definition of the closure operators):  
\begin{Verbatim}  
_→β_ = _→C_ _▹β_
_→β*₀_ = Star _→β_
_≃β_ = EqClosure (_∼α_ ∪ _→β_)
\end{Verbatim}
All of three definitions have type: \Verb|Λ → Λ → Set|. 

Let \binrel\ be any binary relation on $\Lambda$. 
Its \sourcelink{Section2/StarClosure.html}{reflexive and transitive closure},$\Starg\,\binrel$, is inductively defined by the next clauses: (i) \Star{\binrel}{M}{M} for any $M$, and; (ii) if $\binrel\,M\,N$ and \Star{\binrel}{N}{P} then \Star{\binrel}{M}{P}.
Its \sourcelink{Section2/EqClosure.html}{equivalence closure}, $\mathsf{EqClosure}\,\binrel$, is defined as the composition of the symmetric closure of \binrel\ followed by its star closure, 
where the \sourcelink{Section2/SymClosure.html}{symmetric closure} of \binrel\ is the union of \binrel\ and its inverse. 
Its contextual or \sourcelink{Section2/CxtClosure.html}{compatible-with-the-syntax closure},$\_{\rightarrow}\mathsf{C}\_\,\binrel$, is defined as usual, e.g., see \cite[Chapter~3]{barendregt84}.

Next we have some properties about $\beta$-reduction. 
Let \Verb|_𝒮_| be any binary relation on the $\lambda$-terms. We shall say \Verb|_𝒮_| is \sourcelink{Section2/AlphaComm.html}{$\alpha$-commutative},
written \Verb|CommAlpha _𝒮_|, iff: 
\begin{Verbatim}
∀ {M N P} → M ∼α N → N 𝒮 P → ∃ λ Q → M 𝒮 Q × Q ∼α P
\end{Verbatim}
Then we have:
\begin{lemma}
\begin{enumerate}[ref=\thelemma.(\roman{enumi}),itemsep=0pt]

\myitem{Section2/ManyStepCommAlpha.html}
\label{lemma:manyStepCommAlpha}
\Verb|CommAlpha _→β*₀_|

\myitem{Section2/CompatConvSub.html}
\label{lemma:compatConvSub}
\Verb|∀ {M N σ} → M ≃β N → M ∙ σ ≃β N ∙ σ|

\myitem{Section2/InversionManyStep.html}
\label{lemma:key}
\Verb|∀ {x A₁ B₁ C} → Π[ x ∶ A₁ ] B₁ →β*₀ C|\\
\Verb|→ ∃₂ λ A₂ B₂ → C ≡ Π[ x ∶ A₂ ] B₂ × A₁ →β*₀ A₂ × B₁ →β*₀ B₂|
\end{enumerate}
\end{lemma}

\noindent \cref{lemma:key} is sort of an inversion lemma and it will become useful in the proof of SR.

\section{Confluence}
\label{sec:confluence}

\begin{figure}[t]
    \centering
    \begin{subfigure}[b]{0.24\textwidth}        
    	\centering
     	\resizebox{0.8\linewidth}{!}{
            \begin{tikzpicture}
                \coordinate (M) at (0,1);
                \coordinate (N) at (-1,0);
                \coordinate (P) at (1,0);
                \coordinate (Q) at (0,-1);
                
                \draw[->] (M) -- (N);
                \draw[->] (M) -- (P);
                \draw[->,dashed] (N) -- (Q);
                \draw[->,dashed] (P) -- (Q);
            \end{tikzpicture}
        }
        \caption{}
        \label{subfig:diamond}        
    \end{subfigure}      
    \begin{subfigure}[b]{0.24\textwidth}
    	\centering
		\resizebox{0.8\linewidth}{!}{
            \begin{tikzpicture}
\draw[->] (0,1)-- (-1.0517548805800188,0.23999283706445115);
\draw[->,dashed] (-1.0517548805800188,0.23999283706445115)-- (-0.653955247855192,-0.995140624850181);
\draw[dashed] (-0.653955247855192,-0.995140624850181)-- (0.6436533264609949,-0.9984879220201989);
\draw[<-,dashed] (0.6436533264609949,-0.9984879220201989)-- (1.0478198967568657,0.23457679647291607);
\draw[<-] (1.0478198967568657,0.23457679647291607)-- (0,1);
            \end{tikzpicture}
        }
        \caption{}
        \label{subfig:pentagon}        
    \end{subfigure}      
    \begin{subfigure}[b]{0.24\textwidth}
    	\centering
     	\resizebox{0.8\linewidth}{!}{
            \begin{tikzpicture}
                \coordinate (M) at (-0.6,0.6);
                \coordinate (N) at (1,0.6);
                \coordinate (P) at (-1,-0.6);
                \coordinate (Q) at (0.6,-0.6);
                
                \draw[->] (M) -- (N);
                \draw (M) -- (P);
                \draw[dashed] (N) -- (Q);
                \draw[->,dashed] (P) -- (Q);
            \end{tikzpicture}
        }
        \vspace{0.5pt}
        \caption{}
        \label{subfig:comm}
    \end{subfigure}     
    \begin{subfigure}[b]{0.24\textwidth}
    	\centering
		\resizebox{0.8\linewidth}{!}{
            \begin{tikzpicture}
                \coordinate (M) at (0.15,1.35);
                \coordinate (N) at (1,0.5);
                \coordinate (P) at (-0.25,-0.75);
                \coordinate (Q) at (-0.75,-0.75);
                \coordinate (R) at (-1.35,-0.15);

                \draw[->] (M) -- (N);
                \draw[->>] (M) -- (R);   
                \draw[->,dashed] (R) -- (Q);                             
                \draw[->>,dashed] (N) -- (P);
                \draw[dashed] (P) -- (Q);
            \end{tikzpicture}
        }
        \caption{}
        \label{subfig:strip}
    \end{subfigure}
    \caption{Properties of Binary Relations}
\end{figure}
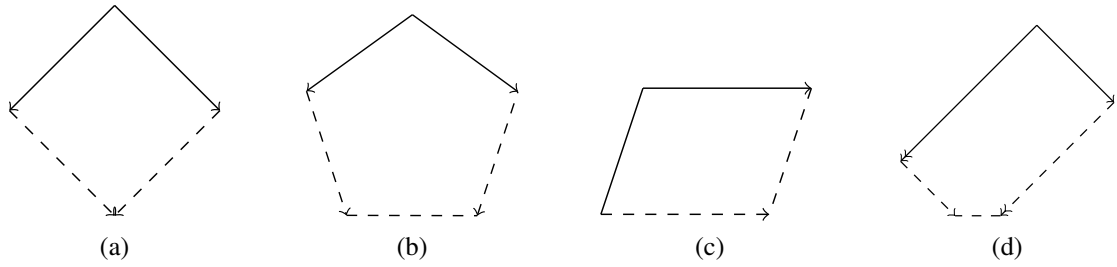

Following the literature 
we define the \sourcelink{Section3/Diamond.html}{diamond property} by:
\begin{Verbatim}
Diam _𝒮_ = ∀ {M N P} → M 𝒮 N → M 𝒮 P → ∃ λ Q → N 𝒮 Q × P 𝒮 Q
\end{Verbatim}
\Cref{subfig:diamond} illustrates it. Vertices are $\lambda$-terms and edges relationship between them. An arrow indicates that the source vertex, say $M$, is related to the target vertex, say $N$, under such relation, i.e., $M\,\binrel\,N$. A line without the arrow head
means that the vertices are $\alpha$-convertible.

In some of the proofs related to the properties above, straight lines will indicate hypotheses while dashed lines will indicate thesis. Particularly, endpoints of thesis will denote $\lambda$-terms that must be found.

If $\mathsf{Star}\,\binrel$ satisfies the diamond property, then we will say \binrel\ is \textit{confluent} in the traditional sense.

The proof of confluence of $\beta$-reduction by Tait and Martin-L\"of can be summarized in two steps.
The first one is a general result about relations. It consists on establishing that the diamond property is preserved by the star closure.
The second step is to find a suitable reduction relation, namely parallel reduction and written \parallelg, such that both it satisfies the diamond property (basing upon the previous result) and its star closure coincides with many-step $\beta$-reduction. Then, by virtue of the first part that it follows $\beta$-reduction is confluent.

Let \sourcelink{Section3/ParallelRed.html}{parallel reduction} be defined below similar to \cite[p.~60]{barendregt84} 
but with some extra rules for our syntax:
\begin{center}
\small
\AxiomC{}
\UnaryInfC{\parallel{k}{k}}
\bottomAlignProof
\DisplayProof
\quad
\AxiomC{}
\UnaryInfC{\parallel{x}{x}}
\bottomAlignProof
\DisplayProof
\quad
\AxiomC{\parallel{A}{A'}}
\AxiomC{\parallel{M}{M'}}
\BinaryInfC{\parallel{\quant[x:A]M}{\quant[x:A']M'}}
\bottomAlignProof
\DisplayProof
\quad
\AxiomC{\parallel{M}{M'}}
\AxiomC{\parallel{N}{N'}}
\BinaryInfC{\parallel{M \app N}{M' \app N'}}
\bottomAlignProof
\DisplayProof
\quad
\AxiomC{\parallel{M}{M'}}
\AxiomC{\parallel{N}{N'}}
\BinaryInfC{\parallel{(\lambda[x:A]M)\app N}{M'[x:=N']}}
\bottomAlignProof
\DisplayProof
\end{center}
\noindent With $\quant\in\{\lambda,\Pi\}$.
Unfortunately, \parallelg\ does \textit{not} satisfy the diamond property in our formal setting. The following illustration provides a counterexample (we ignore domains in $\lambda$-abstractions which are not relevant for the point about to be made and assume names are ordered: $x_0$, $x_1$, $x_2$ \dots):
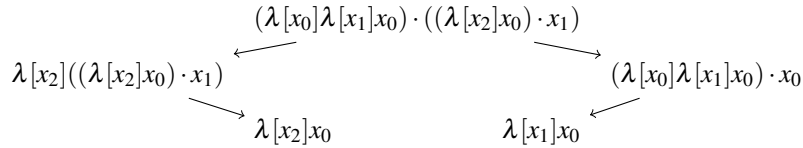
\begin{center}
\begin{tikzpicture}
	\small
    \node (M) at (0,1) {$(\lambda[x_0]\lambda[x_1]x_0)\cdot((\lambda[x_2]x_0) \cdot x_1)$};
    \node[below left = 0.2cm of M] (N) {$\lambda[x_2]((\lambda[x_2]x_0) \cdot x_1)$};
    \node[below right = 0.2cm of M] (P) {$(\lambda[x_0]\lambda[x_1]x_0 )\cdot x_0$}; 
    
    \node[below right = 0.2cm of N] (N') {$\lambda[x_2]x_0$};   
         
    \node[below left = 0.2cm of P] (P') {$\lambda[x_1]x_0$};    
            
    \draw[->] (M) -- (N);
    \draw[->] (M) -- (P);
    \draw[->] (N) -- (N');    
    \draw[->] (P) -- (P');  
\end{tikzpicture}
\end{center}        

The above illustration shows two different reduction strategies for the $\lambda$-term in the top.
On the left hand-side, first we have contracted the left-most redex and then the only left redex, leading to: $\lambda[x_2]x_0$.
On the right hand-side we did similarly, but first we have contracted the right-most redex to obtain: $\lambda[x_1]x_0$. 
We can see both terms differ on the bound variables: $x_2$ and $x_1$. 
One the one hand, the one on the left was fixed during the first reduction, while solving 
${(\lambda[x_1]x_0)[x_0 := (\lambda[x_2]x_0) \cdot x_1)]}$; according to our definition of substitution, the bound variable has to be renamed to 
${X((\iota,x_0:=(\lambda[x_2]x_0)\cdot x_1),\lambda[x_1]x_0)}$, i.e., $x_2$. 
On the other hand, the bound variable $x_1$ did not changed throughout the reduction on the right;
in the first reduction it was not affected at all, and in the second one it was renamed to itself, i.e., ${X((\iota,x_0:=x_0),\lambda[x_1]x_0)}$. 

A solution to the problem above is to add a rule to \parallelg\ to perform an additional $\alpha$-conversion step.
This path was taken in \cite{copello2017} for the pure $\lambda$-calculus.
Nevertheless, a definition of parallel reduction without $\alpha$-conversion was also considered in order to simplify some of the proofs.
As a result, some lemmas ended up duplicated.

In this work we shall explore the alternative path of generalizing confluence up to $\alpha$-conversion and working with a single definition of \parallelg\ that does not mention $\alpha$-conversion right from the start. As we shall see, what we believe is an interesting theory of \textit{$\alpha$-commutative relations} comes up,
and which makes the development cleaner and, at the same time, it saves us from repeating ourselves. 
Besides, we will follow Takahashi's revision of Tait and Martin-L\"of's proof instead of the the original proof.

So, 
to begin with,
we define the 
\sourcelink{Section3/Pentagon.html}{diamond property up to $\alpha$-conversion} or the pentagon property by:
\begin{small}
\begin{Verbatim}
Pent _𝒮_ = ∀ {M N P} → M 𝒮 N → M 𝒮 P → ∃₂ λ Q₁ Q₂ → N 𝒮 Q₁ × P 𝒮 Q₂ × Q₁ ∼α Q₂ 
\end{Verbatim}
\end{small}
Now in this sense we can see that both reductions in the previous example are actually confluent. 

\subsection{Preservation of the Pentagon Property under Star: First Part}
\label{sec:firstPartConfl}

Let us begin by the proving that the diamond property up to $\alpha$-conversion is preserved by the star closure (\cref{lemma:starPreservesPent}). 
To do so, we shall need the following preparatory results.
Let \binrel\ be any $\alpha$-commutative binary relation.
Then we have that \textit{$\alpha$-commutativity} is preserved by the star closure (in what follows we shall write $M \rightarrow_\binrel N$ and $M \twoheadrightarrow_\binrel N$ for $\binrel\,M\,N$ and \Star{\binrel}{M}{N} respectively):
\begin{linklemma}{Section3/StarAlphaComm.html}
\label{lemma:commutativity}
\normalfont
\Verb|CommAlpha (Star 𝒮)|
\end{linklemma}

\begin{proof}
We are given two derivations \alphaconv{M}{N} and \StarAux{N}{P}
and we have to prove there is some $Q$ such that \StarAux{M}{Q} and \alphaconv{Q}{P}. 
We proceed by induction on the derivation of \StarAux{N}{P}. 
\begin{itemize}[itemsep=0pt]
\item Case $P=N$. Then $Q=M$ works fine (see \cref{subfig:comm}). 

\item Case \StarAux{N}{P} follows from \RelAux{N}{P'} and \StarAux{P'}{P} for some $P'$. 
By the hypothesis we have there is some $Q'$ such that \RelAux{M}{Q'} and \alphaconv{Q'}{N'}, thus by the IH we obtain there is some $Q''$ such that \StarAux{Q'}{Q''} and \alphaconv{Q''}{P}. 
Then the $Q$ we are after is $Q''$. 
\qedhere
\end{itemize}
\end{proof}

Next we will define the so-called \textit{strip} property to handle the inductive step in \cref{lemma:starPreservesPent}. The property is illustrated in \cref{subfig:strip}, where a double-headed arrow indicates relationship under the star closure of any  given binary relation \Verb|𝒮|. 
So, we say \Verb|𝒮| satisfies the \sourcelink{Section3/Strip.html}{strip property}, and write it \Verb|Strip 𝒮|, if and only if:
\begin{small}
\begin{Verbatim}
∀ {M N R} → Star 𝒮 M N → 𝒮 M R → ∃₂ λ T₁ T₂ → 𝒮 N T₁ × Star 𝒮 R T₂ × T₁ ∼α T₂
\end{Verbatim}
\end{small}

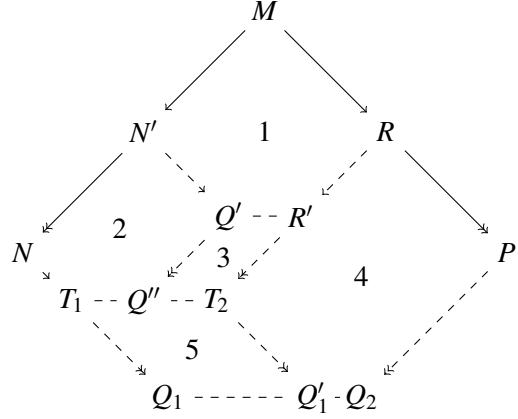
\begin{figure}[t]
    \centering
    
    \begin{tikzpicture}[scale=3.2]       
        \node (M) at (0,1) {$M$};
        \node (N') at (-0.5,0.5) {$N'$};
        \node (Q') at (-0.15,0.15) {$Q'$};
        \node (Q'') at (-0.5,-0.2) {$Q''$};
        \node (N) at (-1,0) {$N$};
        \node (T1) at (-0.8,-0.2) {$T_1$};
        \node (R) at (0.5,0.5) {$R$};
        \node (R') at (0.15,0.15) {$R'$};
        \node (P) at (1,0) {$P$};
        \node (T) at (-0.4,-0.6) {$Q_1$};
        \node (P') at (0.4,-0.6) {$Q_2$};
        \node (S) at (-0.2,-0.2) {$T_2$};
        \node (S') at (0.2,-0.6) {$Q_1'$};

        \draw[->] (M) -- (N');
        \draw[->>] (N') -- (N);
        \draw[->,dashed] (N') -- (Q');
        \draw[->,dashed] (N) -- (T1);
        \draw[dashed] (T1) -- (Q'');
        \draw[->>,dashed] (T1) -- (T);
        \draw[dashed] (T) -- (S');
        \draw[->] (M) -- (R);
        \draw[->>] (R) -- (P);
        \draw[->,dashed] (R) -- (R');
        \draw[->>,dashed] (P) -- (P');
        \draw[->>,dashed] (R') -- (S);
        \draw[dashed] (Q') -- (R');
        \draw[->>,dashed] (Q') -- (Q'');
        \draw[dashed] (S) -- (Q'');
        \draw[->>,dashed] (S) -- (S');
        \draw[dashed] (S') -- (P');
        
        \node at (0,0.5) {1};
        \node at (-0.6,0.1) {2};
        \node at (-0.17,-0.02) {3};
        \node at (0.4,-0.1) {4};
        \node at (-0.3,-0.4) {5};
        
    \end{tikzpicture}
    \caption{Proofs of the Strip Lemma and Preservation of the Pentagon Property}
    \label{fig:proofClosureDiamondStar}
\end{figure}

\noindent Then we have: 
\begin{linklemma}{Section3/StripLemma.html}
\Verb|Pent 𝒮 → Strip 𝒮|
\label{lemma:strip}
\end{linklemma}

\begin{proof}
By structural induction on some given derivation of \StarAux{M}{N} (we are also given \RelAux{M}{R} as a hypothesis). 
We have to find some $\alpha$-convertible terms $T_1$ and $T_2$ such that \RelAux{N}{T_1} and \StarAux{R}{T_2}. 
\begin{itemize}[itemsep=0pt]
\item Case $M=N$. Then the whole structure or shape collapses and $T_1 = T_2 = R$. 

\item Case \StarAux{M}{N} follows from \RelAux{M}{N'} and \StarAux{N'}{N}. We can use the diagram chase suggested by \cref{fig:proofClosureDiamondStar} to navigate through the proof, focusing only on the areas with numbers 1, 2 and 3. 
These numbers also indicate the sequence of steps in the proof. 
First (1), by confluence we have that there are some $\alpha$-convertible terms $Q'$ and $R'$ such that \RelAux{N'}{Q'} and \RelAux{R}{R'}. 
Second (2), by the IH we have our first corner, $T_1$, i.e., \RelAux{N}{T_1}, and some $\alpha$-convertible term $Q''$ such that \StarAux{Q'}{Q''}. 
And thirdly (3), by \cref{lemma:commutativity} we have that the star closure of \binrel\ is also $\alpha$-commutative (\cref{subfig:comm}), thus we can reach to the last corner $T_2$ from $R'$, which, since it is $\alpha$-convertible to $Q''$, it is also $\alpha$-convertible to $T_1$.
\qedhere
\end{itemize}
\end{proof}

Having established the previous results we can 
complete
the first step in the proof of confluence:

\begin{linklemma}{Section3/StarPreservesPent.html}
\label{lemma:starPreservesPent}
\Verb|Pent 𝒮 → Pent (Star 𝒮)|
\end{linklemma}

\begin{proof}
The proof is carried out by structural induction on the derivation of 
\StarAux{M}{N} (we are also given \StarAux{M}{P} as a hypothesis).
The only interesting case is the inductive one, i.e., when \StarAux{M}{N} follows from \RelAux{M}{N'} and \StarAux{N'}{N}.
There, we continue by case analysis on the structure of \StarAux{M}{P}. If $M=P$, then the result is trivial.
If \StarAux{M}{P} follows from \RelAux{M}{R} and \StarAux{R}{N}, then we proceed as follows. 
First, and similarly to the previous lemma, we do steps (1), (2) and (3) using \cref{fig:proofClosureDiamondStar} as a guide.
Now we have to take a look at the whole graph. 
The fourth step (4) is to use the IH; we obtain $Q_1'$ and $Q_2$. 
Finally (5), by \cref{lemma:commutativity} we can commute \alphaconv{T_1}{T_2} and \StarAux{T_2}{Q_1'} to obtain $Q_1$ as desired. 
\qedhere
\end{proof}

\subsection{Properties of Parallel Reduction: Second Part}
\label{sec:secondPartConfl}

Now we have to prove that parallel reduction satisfies the pentagon property. 
To this end, first we will need to establish that \Starg{\parallelg} coincides with \manystepg, and that \parallelg\ is compatible with substitutions and commutes with $\alpha$-conversion (the proofs are adapted from \cite{copello2017}): 
\begin{lemma}
\label{lemma:propertiesParallel}
\begin{enumerate}[ref=\thelemma.(\roman{enumi}),itemsep=0pt]

\myitem{Section3/EquivBetaPar.html}
\label{lemma:eqBetaParStar}
\Verb|_→β*₀_ ⇔ Star _⇉_|

\myitem{Section3/CompatParSub.html}
\label{lemma:compatParSubUnary}
\Verb|∀{M M' N N' x}→ M ⇉ M' → N ⇉ N' → ∃ λ P → M [ x := N ] ⇉ P × P ∼α M' [ x := N' ]|

\myitem{Section3/CommPent.html}
\label{lemma:pentComm}
\Verb|CommAlpha _⇉_|

\end{enumerate}   
\end{lemma}

Next we will define Takahashi's star operator (\Verb|_⋆|) by recursion on the syntax. Given a $\lambda$-term, $M^*$ is the result of contracting all redexes in $M$ simultaneously. We will use informal notation to avoid some technicalities on Agda's function definitions and refer the interested reader to the source: 
\begin{gather*}
x^* = x \quad k^* = k \quad (\lambda[x:A]M)^* = \lambda[x:A^*](M^*) \quad 
(\Pi[x:A]B)^* = \Pi[x:A^*](B^*)\\
(M \cdot N)^* = \begin{cases}M'^*[x:=N^*] &\text{if}\ M=\lambda[x:A]M'\\M^* \cdot N^*\ & \text{otherwise}\end{cases}
\end{gather*}

The following lemma is also due to Takahashi:

\begin{linklemma}{Section3/TakahashiLemma.html}
\Verb|∀ {M N} → M ⇉ N → ∃ λ P → N ⇉ P × P ∼α M ⋆|
\label{lemma:takahashi}
\end{linklemma}

\begin{proof}
By structural induction on $M$ and subordinate case analysis on the derivation of \parallel{M}{N}. We only show the most interesting case, i.e., when $M$ is a redex: ${M=(\lambda[x:A] M_0) \cdot M_1}$.
Our goal is to find some $P$ such that \parallel{N}{P} and \alphaconv{P}{M_0^*[x:=M_1^*]}. 
There are two cases as to the derivation of 
\parallel{(\lambda[x:A] M_0) \cdot M_1}{N}:
\begin{itemize}[itemsep=0pt]
     \item If the last rule applied is \Verb|⇉·| then $N=(\lambda[x:B]N_0) \cdot N_1$, and we have \parallel{\lambda[x:A]M_0}{\lambda[x:B]N_0} and \parallel{M_1}{N_1}. Also, the first premise can only follow from \Verb|⇉λ| so we have \parallel{A}{B} and \parallel{M_0}{N_0}.
     By the IH we have that there are some $P_0$ and $P_1$ such that \parallel{N_0}{P_0} and \alphaconv{P_0}{M_0^*}, and \parallel{N_1}{P_1} and \alphaconv{P_1}{M_1^*}. Then $P=P_1[x:=P_2]$; by the \parallelbeta\ rule we have \parallel{(\lambda[x:B]N_0) \cdot N_1}{P_0[x:=P_1]} and by \cref{lemma:compatAlphaSubUnary} we get \alphaconv{P_0[x:=P_1]}{M_0'^*[x:=M_1^*]}.   

    \item If the last rule applied is \Verb|⇉β| then $N=N_0[x:=N_1]$, and we have \parallel{A}{B}, \parallel{M_0}{N_0} and \parallel{M_1}{N_1}.
    First, by the IH we obtain \parallel{N_0}{P_0} and \alphaconv{P_0}{M_0^*}, and \parallel{N_1}{P_1} and \alphaconv{P_1}{M_1^*} for some $P_0$ and $P_1$. Next, by \cref{lemma:compatParSubUnary} we have \parallel{P_0[x:=P_1]}{Q} and \alphaconv{Q}{N_0[x:=N_1]} for some $Q$. Let $P=Q$; 
    by \cref{lemma:compatAlphaSubUnary} we have \alphaconv{P_0[x:=P_1]}{M_0^*[x:=M_1^*]} and by transitivity using the previous result \alphaconv{P}{M_0^*[x:=M_1^*]}.
    \qedhere
\end{itemize}
\end{proof}

\begin{linklemma}{Section3/PentPar.html}
\label{lemma:pentPar}
\Verb|Pent _⇉_|
\end{linklemma}

\begin{proof}
Immediate by using Lemmas~\ref{lemma:takahashi}~and~\ref{lemma:alphaEquivalence}.\footnote{In \cite{copello2017}, Lemmas \ref{lemma:pentComm} and \ref{lemma:pentPar} were proven twice, once for \parallelg\ as is and another including $\alpha$-conversion.}
\end{proof}

Once we know the previous facts about \parallelg, the CR theorems follow easily:

\begin{theorem}[Church-Rosser]
\normalfont
\label{theorem:CR}
\begin{enumerate}[ref=\thetheorem.(\roman{enumi}),itemsep=0pt]
\thmitem{Section3/CR1.html} \Verb|Pent _→β*₀_| 
\thmitem{Section3/CR2.html} \Verb|∀ {M N} → M ≃β N → ∃₂ λ P₁ P₂ → M →β*₀ P₁ × N →β*₀ P₂ × P₁ ∼α P₂|
\end{enumerate}
\end{theorem}

\begin{proof}
(i) follows by using \cref{lemma:pentComm,lemma:starPreservesPent}, and then \cref{lemma:eqBetaParStar}. (ii) is by structural induction on the derivation of \conv{M}{N} and by using (i) and \cref{lemma:manyStepCommAlpha}.\qedhere
\end{proof}

To finish this section, we present a list of equations showing that $\lambda$-calculus is consistent; some of them will be used later in this development:
\begin{lemma}
\begin{enumerate}[ref=\thelemma.(\roman{enumi}),itemsep=0pt]

\myitem{Section3/Absurd1.html} \Verb|∀ {x s A B} → ¬(Π[ x ∶ A ] B ≃β c s)|
\label{lemma:absurdEquationProdSorts}

\myitem{Section3/Absurd2.html} \Verb|∀ {x s} → ¬(v x ≃β c s)|
\label{lemma:absurdEquationVarSort}

\myitem{Section3/Absurd3.html} \Verb|∀ {y x A B} → ¬(v y ≃β Π[ x ∶ A ] B)|
\label{lemma:absurdEquationVarProd}
\end{enumerate}
\end{lemma}

\section{The Pure Type Systems and Subject Reduction}
\label{sec:sr}

A PTS is specified by a triple $(\mathcal{S}, \mathcal{A}, \mathcal{R})$ of sorts, axioms and rules respectively. We shall use $\mathcal{C} = \mathcal{S}$, i.e., we will instantiate the constants in the previous framework with the set of sorts of the particular PTS. 
The rules of the \sourcelink{Section4/Typing.html}{typing judgment} (\verb|_⊢ₛ_:_|) and \sourcelink{Section4/ValidCxt.html}{valid contexts} (\verb|_okₛ|) are mutually defined and shown in \cref{fig:ptsInformal} using informal notation.\footnote{In \cite{urciuoli2025}, two different presentations of the PTS were introduced. One using infinitary branching trees and a modified rule for applications, and another more standard. The latter puts an ``s'' as a subscript to the entailment symbol and it is the one we are going to use here, nevertheless, we remark that both presentations are (extensionally) equivalent.
} 
$\Gamma,x:A$ is definitionally equal to $(x,A) :: \Gamma$.

\begin{figure}[t]
    \centering
    \small
    
    \AxiomC{}
    \LeftLabel{\typerule{nil}} 
    \UnaryInfC{$\ok{[]}$}
    \bottomAlignProof
    \DisplayProof
    \quad
    \AxiomC{\ok{\Gamma}}
    \AxiomC{$\Gamma \vdashstd A : s$}
    \LeftLabel{\footnotesize\typerule{cons}} 
    \RightLabel{($x \not\in \dom{\Gamma}$)}
    \BinaryInfC{\ok{\Gamma,x:A}}
    \bottomAlignProof
    \DisplayProof
    
    \bigskip
    
    \AxiomC{$\ok{\Gamma}$}
    \LeftLabel{\typerule{sort}}
    \RightLabel{($\mathcal{A}\,s_1\,s_2$)}
    \UnaryInfC{$\Gamma \vdashstd s_1 : s_2$}
    \bottomAlignProof
    \DisplayProof
    \quad    
    \AxiomC{\ok{\Gamma}}
    \LeftLabel{\typerule{var}}
    \RightLabel{($(x,A) \in \Gamma$)}
    \UnaryInfC{$\Gamma \vdashstd x : A$}
    \bottomAlignProof    
    \DisplayProof
    
    \bigskip
    
    \AxiomC{$\Gamma \vdashstd A : s_1$}
    \AxiomC{$\Gamma,y:A \vdashstd B[x:=y] : s_2$}
    \LeftLabel{\typerule{prod}}
    \RightLabel{\small$\begin{cases}\mathcal{R}\,s_1\,s_2\,s_3\\y \not\in\fv{B}-x\end{cases}$}
    \BinaryInfC{$\Gamma \vdashstd \Pi[x:A]B : s_3$}
    \bottomAlignProof
    \DisplayProof
    
    \bigskip

    \AxiomC{$\Gamma \vdashstd A : s_1$}
    \AxiomC{$\Gamma, z : A \vdashstd M[x:=z] : B[y:=z]$}           
    \AxiomC{$\Gamma, z : A \vdashstd B[y:=z] : s_2$}                   
	\LeftLabel{\typerule{abs}}
    \RightLabel{$\begin{cases}\mathcal{R}\,s_1\,s_2\,s_3\\z\not\in\fv{M}-x\\z\not\in\fv{B}-y\end{cases}$}
    \TrinaryInfC{$\Gamma \vdashstd \lambda[ x : A ] M : \Pi[ y : A ] B$}
    \DisplayProof
    
    \bigskip
        
    \AxiomC{$\Gamma \vdashstd M : \Pi[ x : A ] B$}
    \AxiomC{$\Gamma \vdashstd N : A$}
    \LeftLabel{\typerule{app}}
    \BinaryInfC{$\Gamma \vdashstd M \cdot N : B[x:=N]$}
    \DisplayProof

    \bigskip
                
    \AxiomC{$\Gamma \vdashstd M : A$}    
    \AxiomC{$\Gamma \vdashstd B : s$}
    \LeftLabel{\typerule{conv}}
	\RightLabel{($\conv{A}{B}$)}
    \BinaryInfC{$\Gamma \vdashstd M : B$}
    \DisplayProof
    \caption{Pure Type Systems}
    \label{fig:ptsInformal}
\end{figure}

The \textit{sorts} ($\mathcal{S}$) are the classifiers for types, e.g., in the $\lambda$-cube there two sorts: $*$ and $\square$ \cite{barendregt1991}.
The meaning of $*$ depends on the specific calculus; for instance, 
in $\lambda$C (CC), $*$ is the type of \textit{propositions}, while in $\lambda$P (LF) it usually denotes the type of sets and \textit{judgments}.
As to $\square$, it classifies some types, including $*$, into what are called \textit{kinds}.
The \textit{axioms} ($\mathcal{A}$) describe the relationship between the sorts, e.g., $* : \square$ in all systems of the $\lambda$-cube. 
Finally, the \textit{rules} ($\mathcal{R}$) are some conditions that moderate the construction of function types; 
for instance, in $\lambda$P, the domain in functions can only be types with sort $*$; using set-theoretic notation, $\mathcal{R}_{\lambda\text{P}} = \{ (*, *, *) , (*, \square, \square) \}$.

The next results were proven in \cite{urciuoli2025}:
\begin{lemma}
\begin{enumerate}[ref=\thelemma.(\roman{enumi}),label=(\roman{enumi}),itemsep=0pt]

\item 
\sourcelink{Section4/CxtValidity.html}{\textit{(Context Validity)}} 
\Verb|∀ {Γ M A} → Γ ⊢ₛ M ∶ A → Γ okₛ|
\label{lemma:contextValidity}

\item 
\sourcelink{Section4/InversionSorts.html}{\textit{(Inversion of Sorts)}}
\Verb|∀ {Γ s A} → Γ ⊢ₛ c s ∶ A → ∃ λ t → Γ okₛ × 𝒜 s t × A ≃β c t|
\label{lemma:inversionTypeSorts}

\item 
\sourcelink{Section4/InversionVars.html}{\textit{(Inversion of Variables)}}
\Verb|∀{Γ x A} → Γ ⊢ₛ v x ∶ A → ∃ λ B → Γ okₛ × (x , B) ∈ Γ × A ≃β B|
\label{lemma:inversionTypeVar}

\item 
\sourcelink{Section4/InversionProd.html}{\textit{(Inversion of $\Pi$-types)}}
\Verb|∀ {Γ x A B C} → Γ ⊢ₛ Π[ x ∶ A ] B ∶ C → ∃₄ λ s₁ s₂ s₃ y|\\
\Verb|→ ℛ s₁ s₂ s₃ × Γ ⊢ₛ A ∶ c s₁ × y ∉ fv B - x|\\
\Verb|× Γ ‚ y ∶ A ⊢ₛ B [ x := v y ] ∶ c s₂ × C ≃β c s₃|

\item 
\sourcelink{Section4/InversionLam.html}{\textit{(Inversion of $\lambda$-abstractions)}}
\Verb|∀ {Γ x A M C} → Γ ⊢ₛ λ[ x ∶ A ] M ∶ C → ∃₄ λ s x' y B|\\
\Verb|→ y ∉ fv M - x × y ∉ fv B - x' × Γ ‚ y ∶ A ⊢ₛ M [ x := v y ] ∶ B [ x' := v y ]|\\
\Verb|× Γ ⊢ₛ Π[ x' ∶ A ] B ∶ c s × C ≃β Π[ x' ∶ A ] B|
\label{lemma:inversionTypeAbs}

\item 
\sourcelink{Section4/InversionApp.html}{\textit{(Inversion of Applications)}}
\Verb|∀ {Γ M N C} → Γ ⊢ₛ M · N ∶ C|\\ 
\Verb|→ ∃₃ λ x A B → Γ ⊢ₛ M ∶ Π[ x ∶ A ] B × Γ ⊢ₛ N ∶ A × C ≃β B [ x := N ]|
\label{lemma:inversionTypeApp}

\item 
\sourcelink{Section4/Weakening.html}{\textit{(Weakening)}}
\Verb|∀{Γ x s M A B}→ Γ ⊢ₛ M ∶ B → x ∉ dom Γ → Γ ⊢ₛ A ∶ c s → Γ ‚ x ∶ A ⊢ₛ M ∶ B|
\label{lemma:weakening}

\item 
\sourcelink{Section4/SyntacticValidity.html}{\textit{(Syntactic Validity)}}
\Verb|∀ {Γ M A} → Γ ⊢ₛ M ∶ A → ∃ λ s → A ≡ c s ⊎ Γ ⊢ₛ A ∶ c s|

\item 
\sourcelink{Section4/ClosAlpha.html}{\textit{(Closure under $\alpha$-conversion)}}
\Verb|∀ {Γ M N A} → M ∼α N → Γ ⊢ₛ M ∶ A → Γ ⊢ₛ N ∶ A|

\item \sourcelink{Section4/Cut.html}{\textit{(Cut)}}
\Verb|∀{Γ M N A B x}→ Γ ‚ x ∶ A ⊢ₛ M ∶ B → Γ ⊢ₛ N ∶ A → Γ ⊢ₛ M [ x := N ] ∶ B [ x := N ]|
\end{enumerate}
\end{lemma}

\subsection{The Proof of Subject Reduction}

To prove SR first we need to prove product injectivity:

\begin{lemma}[\sourcelink{Section4/ProductInjectivity.html}{Product Injectivity}]
\label{lemma:productInjectivity}
\Verb|∀ {x₁ x₂ A₁ A₂ B₁ B₂} → Π[ x₁ ∶ A₁ ] B₁ ≃β Π[ x₂ ∶ A₂ ] B₂|\\
\Verb|→ A₁ ≃β A₂ × ∀ y → B₁ [ x₁ := v y ] ≃β B₂ [ x₂ := v y ]|
\end{lemma}

\begin{proof}
By 
the second CR theorem
we know there are some $\alpha$-convertibles $\lambda$-terms $C_1$ and $C_2$ such that \manystep{\Pi[ x_i ∶ A_i ] B_i}{C_i} for $i=1,2$.
By \cref{lemma:key} we also know there are some $A_i'$ and $B_i'$ such that $C_i \equiv \Pi[x_i:A_i']B_i'$ with \manystep{A_i}{A_i'} and \manystep{B_i}{B_i'} for $i=1,2$. 
Now, first notice that, by the rule of products, \alphaconv{C_1}{C_2} must follow from (a) \alphaconv{A_1'}{A_2'} and (b) $B_1'[x_1:=y']\equiv B_2'[x_2:=y']$ for some $y'$ fresh for $\fv{B_1'} - x_1$ and $\fv{B_2'} - x_2$
As to the the first goal, i.e., \conv{A_1}{A_2}, it follows by using the following sequence of conversions: \manystep{A_1}{A_1'}, \alphaconv{A_1'}{A_2'} and $A_2'\ {\leftarrow}\beta{*}_0\ A_2$. 
As to the second goal,
let $y$ be any variable. We proceed as follows:
\begin{multline*}
B_1 [ x_1 := y ] \stackrel{(1)}{\convgspc} B_1' [ x_1 := y ] \stackrel{(2)}{\equiv} B_1' [ x_1 := y' ] [ y' := y ] \stackrel{(3)}{\equiv} B₁' [ x₁ := y' ] [ y' := y ] \\
\stackrel{(4)}{\equiv} B₂' [ x₂ := y' ] [ y' := y ] \stackrel{(5)}{\equiv}  B₂' [ x₂ := y ] \stackrel{(6)}{\convgspc} B₂ [ x₂ := y ]
\end{multline*}
In (1) and (6) we have applied \cref{lemma:compatConvSub}; in (2) and (5), \cref{lemma:chi} and \cref{lemma:expandUpd}, and finally; in (3) and (4) congruence of $\_{\equiv}\_$ on (b) with the substitution operation. 
\end{proof}

Next we will need some preparatory definitions.
Since typing derivations may move redexes between the subjects and the contexts we will take the same path as McKinna and Pollack and extend reductions to the latter;
we will then state and prove simultaneously two lemmas, 
one that shows context validity is preserved by reduction (Theorem~\ref{theo:sr1}),
and another for the soundness of both context and subject reduction w.r.t. typing, i.e., SR (Theorem~\ref{theo:sr2}).
Also, because of sometimes we will need to invoke the induction hypothesis with a single reduction, either in the context or in the subject, then we will require that both relations be reflexive.
Thus, below we define \sourcelink{Section4/ReflexiveBeta.html}{reflexive $\beta$-reduction} (\Verb|_→β₌_|) and its \sourcelink{Section4/ReflexiveBetaCxt.html}{extension to contexts pointwise} (\Verb|_→→β₌_|): 
\begin{Verbatim}
_→β₌_ = _≡_ ∪ _→β_
_→→β₌_ = Pointwise (λ (x , A) (y , B) → x ≡ y × A →β₌ B)
\end{Verbatim}

\newcommand{\reflstepcxt}[2]{\ensuremath{#1\ {\rightarrow}{\rightarrow}_{\beta_=}\ #2}}
\newcommand{\reflstep}[2]{\ensuremath{#1 \rightarrow_{\beta_=} #2}}

We can now state and prove the SR theorem:

\begin{theorem}[Subject Reduction]
\label{theo:sr}
\begin{enumerate}[ref=\thetheorem.(\roman{enumi}),itemsep=0pt]

\thmitem{Section4/SR1.html} 
\label{theo:sr1}
\Verb|∀ {Γ Δ} → Γ okₛ → Γ →→β₌ Δ → Δ okₛ|

\thmitem{Section4/SR2.html} 
\label{theo:sr2}
\Verb|∀ {Γ Δ M N A} → Γ ⊢ₛ M ∶ A → Γ →→β₌ Δ → M →β₌ N → Δ ⊢ₛ N ∶ A|

\end{enumerate}
\end{theorem}

\begin{proof}
By simultaneous induction on 
the derivations of \ok{\Gamma} and \types{\Gamma}{M}{A}, and subordinate case analysis on \reflstepcxt{\Gamma}{\Delta} and \reflstep{M}{N}. 
We only show the most complex case in (ii), i.e.,
when the last rule applied is that of applications and the redex in $M$ being contracted is the outer-most one.
There we have 
$M=(\lambda[x:A_1]M_1) \cdot M_2$ and $N=M_1[x:=M_2]$, 
and \types{\Gamma}{(\lambda[x:A_1]M_1) \cdot M_2}{B_2[y:=M_2]} which
follows from \types{\Gamma}{\lambda[x:A_1]M_1}{\Pi[y:B_1]B_2} and \types{\Gamma}{M_2}{B_1}.
We must derive: \types{\Delta}{M_1[x:=M_2]}{B_2[y:=M_2]}. In order to do so, first we will derive \types{\Delta}{M_1[x:=M_2]}{C} for some term $C$ such that \conv{C}{B_2[y:=M_2]} and then use the conversion rule.
We proceed as follows:
\begin{enumerate}[label=\arabic*.,itemsep=0pt]
	\item By the IH we can convert the contexts on the premises and obtain: (i) \types{\Delta}{\lambda[x:A_1]M_1}{\Pi[y:B_1]B_2}, and; (ii) \types{\Delta}{M_2}{B_1}.
	
	\item By inversion on 1.(i) we know there are some $s_1$, $s_2$, $y'$, $B_2'$ and $z$ fresh for $\fv{M_1}-x$ and $\fv{B_2'}-y'$ such that (i) \types{\Delta}{A_1}{s_1}, (ii) \types{\Delta,z:A_1}{M_1[x:=z]}{B_2'[y':=z]}, (iii) \types{\Delta}{\Pi[y':A_1]B_2'}{s_2}; and; (iv) \conv{\Pi[y:B_1]B_2}{\Pi[y':A_1]B_2'}.
	
	\item By 
    \cref{lemma:productInjectivity}
	on 2.(iv) it follows: (i) \conv{B_1}{A_1} and (ii) \conv{B_2[y:=w]}{B_2'[y':=w]} for all $w$.
	
	\item By syntactic validity on 1.(i) we also obtain \types{\Delta}{\Pi[y:B_1]B_2}{s} for some $s$.
	
	\item By inversion on this last result we know there are some $s_1'$, $s_2'$ and $z'$ fresh $\fv{B_2}-y$ for such that (i) \types{\Delta}{B_1}{s_1'} and (ii) \types{\Delta,z':B_1}{B_2[y:=z']}{s_2'}.

	\item By the conversion rule using 1.(ii), 3.(i) and 2.(i) we have: \types{\Delta}{M_2}{A_1}.
	
	\item By cut on 2.(ii) and 6 and we have \types{\Delta}{M_1[x:=z][z:=M_2]}{B_2'[y':=z][z:=M_2]} which, by \cref{lemma:expandUpd}, equals to \types{\Delta}{M_1[x:=M_2]}{B_2'[y':=M_2]}. Let $C=B_2'[y':=M_2]$; then we have the first part of our goal. 
	
	\item Now, as to \conv{B_2'[y':=M_2]}{B_2[y:=M_2]}, first let ${z''=X'(\fv{B_2'}-y' \append\ \fv{B_2}-y)}$; by \cref{lemma:chiPrime}
	we have that $z''$ is fresh for $\fv{B_2'-y'}$ and $\fv{B_2}-y$.
	Next, by \cref{lemma:compatConvSub} on 3.(ii) we have that \conv{B_2[y:=z][z:=M_2]}{B_2'[y':=z][z:=M_2]}, and finally by \cref{lemma:expandUpd} and symmetry of conversion we obtain the second part of our goal.
	
	\item Finally, to use the conversion rule we have to prove that $C$ is a valid type, i.e., \types{\Delta}{B_2'[y':=M_2]}{s} for some $s$. 
	This can be easily accomplished by using cut on 5.(ii) and 6 in order to obtain that \types{\Delta}{B_2[y:=z'][z':=M_2]}{s_2'}, which by virtue of \cref{lemma:expandUpd} equals to \types{\Delta}{B_2[y:=M_2]}{s_2'}.		
\qedhere
\end{enumerate}
\end{proof}

\section{Consistency}
\label{sec:consistency}

In this section we will prove that for some class of impredicative PTS to be precisely defined and assumed to be normalizing, it follows that there is a type that cannot be inhabited in the empty context.
The proof we present extends the one for CC \cite{coquand1990}.

\subsection{Normal and Neutral Forms}
\label{sec:normalforms}

We say some $\lambda$-term $M$ is in \sourcelink{Section5/nf.html}{normal form}, written \nf{M}, iff it is irreducible:
\begin{Verbatim}
nf M = ∀ {N} → ¬(M →β N)
\end{Verbatim}
Particularly, if $M$ is an iteration of applications blocked by a variable at the head, i.e., $M = x \cdot M_1 \cdot M_2 \dots$ then we say it is in \sourcelink{Section5/ne.html}{neutral form}, and write it \neu{M}. Below we define it by recursion on the syntax:
\begin{Verbatim}
ne (var x) = ⊤; ne (M · N) = ne M × nf N; ne _ = ⊥
\end{Verbatim}

We can say then that some $\lambda$-term $M$ is \sourcelink{Section5/wn.html}{weakly normalizing} iff every sequence of reductions beginning on $M$ ends with some other $\lambda$-term in normal form:
\begin{Verbatim}
wn M = ∃ λ N → nf N × M →β*₀ N
\end{Verbatim}

The next inversion lemmas about \nfg\ follow directly:
\begin{lemma}[Inversion Lemmas for \nfg]
\label{lemma:inversionnf}
\begin{enumerate}[ref=\thelemma.(\roman{enumi}),label=(\roman{enumi}),itemsep=0pt]

\myitem{Section5/InversionAppNf.html} 
\Verb|∀ {M N} → nf (M · N) → nf M × nf N|
\label{lemma:inversionAppNf}

\myitem{Section5/InversionAbsNf.html} 
\Verb|∀ {x A M} → nf (λ[ x ∶ A ] M) → nf A × nf M|
\label{lemma:inversionAbsNf}

\myitem{Section5/InversionProdNf.html} 
\Verb|∀ {x A B} → nf (Π[ x ∶ A ] B) → nf A × nf B|
\end{enumerate}
\end{lemma}

Let $(\mathcal{S},\mathcal{A},\mathcal{R})$ be any PTS which we also identify with $\_\vdash_{\mathsf{s}}\_∶\_$.
Then we say said PTS is \sourcelink{Section5/Normalizing.html}{normalizing} iff every subject and predicate are weakly normalizing:
\begin{Verbatim}
Normalizing = ∀ {Γ M A} → Γ ⊢ₛ M ∶ A → wn M × wn A
\end{Verbatim}
In the code, both the module containing the above definition and the one with the definition of the PTS are parameterized by a triple $(\mathcal{S},\mathcal{A},\mathcal{R})$. The current module then acts as sort of a proxy and opens the others with the supplied triple.

The definitions of normal and neutral forms above naturally reflect our intuition on such concepts. However, they are not entirely fit for some of the upcoming results. Hence, we also define 
\sourcelink{Section5/NormalForm.html}{normal-} and \sourcelink{Section5/NeutralForm.html}{neutral- forms} $\lambda$-terms by (mutual) \textit{induction}, using informal notation, with the following set of rules:
\begin{center}
\AxiomC{}
\UnaryInfC{\Nf{s}}
\bottomAlignProof
\DisplayProof
\quad
\AxiomC{\Nf{A}}
\AxiomC{\Nf{M}}
\BinaryInfC{\Nf{(\lambda[x:A]M)}}
\bottomAlignProof
\DisplayProof
\quad
\AxiomC{\Nf{A}}
\AxiomC{\Nf{B}}
\BinaryInfC{\Nf{(\Pi[x:A]B)}}
\bottomAlignProof
\DisplayProof
\quad
\AxiomC{\Neu{x}{M}}
\UnaryInfC{\Nf{M}}
\bottomAlignProof
\DisplayProof
\quad
\AxiomC{}
\UnaryInfC{\Neu{x}{x}}
\bottomAlignProof
\DisplayProof
\quad
\AxiomC{\Neu{x}{M}}
\AxiomC{\Nf{N}}
\BinaryInfC{\Neu{x}{(M \cdot N)}}
\bottomAlignProof
\DisplayProof
\end{center}
We have also added the head variable to the type in the neutral form predicate for convenience.

Note that the previous predicate exclude some normal-form terms such as $s \app M$ or $(\Pi[x:A]B) \app M$. 
These, however, cannot occur in any typing derivation. 
Since we will only consider well-typed terms, it turns out that our definition is suitable enough.

\subsection{Soundness and Completeness of the Inductive Characterization of Normal Forms}
\label{sec:soundnessnf}

To delegate some of the results to the inductive definition we will require a lemma of completeness, i.e., that \nfg/\neug\ implies \Nfg/\Neg.
As to the opposite direction, i.e., soundness, it turns out that at some point (\cref{lemma:consistencyNf}) we are going to have a derivation of \Nf{A} for some $A$, and it will require us to use some lemma (\cref{lemma:normalTermConvertsToSort}) that follows much more easily by using \nf{A} instead.
Hence, and to close the circle, we will also prove soundness.
 
Soundness follows directly by mutual induction on the structure of the derivations:
\begin{lemma}[Soundness of \Nfg/\Neg]
\begin{enumerate}[ref=\thelemma.(\roman{enumi}),label=(\roman{enumi}),itemsep=0pt]

\myitem{Section5/SoundnessNF.html}
\Verb|∀ {M} → Nf M → nf M|
\label{lemma:soundnessNf}

\myitem{Section5/SoundnessNE.html}
\Verb|∀ {x M} → Ne x M → nf M|
\end{enumerate}
\end{lemma}

Completeness, on the other hand, depends on some results about renamings. Let \IsVarg\ be the predicate on $\Lambda$ that holds iff the subject is a variable. Then we define \sourcelink{Section5/Renaming.html}{renamings} as substitutions mapping any two set of variables:
\begin{Verbatim}
Ren ρ = ∀ x → IsVar (ρ x)
\end{Verbatim}

The next result follows easily:
\begin{linklemma}{Section5/UnaryRen.html}
\Verb|∀ {x y} → Ren (ι ‚ x := v y)|
\label{lemma:updateRenaming}
\end{linklemma}

Then we have some compatibility results about renamings and normal forms:
\begin{lemma}
\begin{enumerate}[ref=\thelemma.(\roman{enumi}),label=(\roman{enumi}),itemsep=0pt]

\myitem{Section5/Equivariance.html}
\Verb|∀ {M ρ} → Ren ρ → nf M → nf (M ∙ ρ)|

\myitem{Section5/EquivarianceNF.html}
\Verb|∀ {M ρ} → Ren ρ → Nf M → Nf (M ∙ ρ)|
\label{lemma:stabilityNfUnderRenaming}

\myitem{Section5/EquivarianceNE.html}
\Verb|∀ {x ρ M} → Ren ρ → Ne x M → ∃ λ y → Ne y (M ∙ ρ)|

\myitem{Section5/AntirenamingNF.html}
\Verb|∀ {M ρ} → Ren ρ → Nf (M ∙ ρ) → Nf M|
\label{lemma:antirenamingNf}

\myitem{Section5/AntirenamingNE.html}
\Verb|∀ {x ρ M} → Ren ρ → Ne x (M ∙ ρ) → ∃ λ y → Ne y M|
\end{enumerate}
\end{lemma}

\begin{proof}
(i) follows by induction on the syntax, using \cref{lemma:updateRenaming} and a corresponding inversion lemma (\cref{lemma:inversionnf}) for each constructor. Both pairs of lemmas (ii) and (iii), and (iv) and (v), are proven simultaneously by  induction on the derivation of the normal form of $M$ and by using \cref{lemma:updateRenaming}.
\end{proof}

Finally, we can prove that \Nfg\ and \Neg\ are complete w.r.t. the intuitive characterization:

\begin{lemma}[Completeness of \Nfg/\Neg]
\label{lemma:completenessNf}
\begin{enumerate}[ref=\thelemma.(\roman{enumi}),label=(\roman{enumi}),itemsep=0pt]
\myitem{Section5/CompletenessNf.html}
\Verb|∀ {Γ M A} → Γ ⊢ₛ M ∶ A → nf M → Nf M|

\myitem{Section5/CompletenessNe.html} 
\Verb|∀ {M Γ A} → Γ ⊢ₛ M ∶ A → ne M → ∃ λ x → Ne x M|
\end{enumerate}
\end{lemma}

\begin{proof}
The proofs are carried out simultaneously by induction on the structure of the typing derivation. 
We will only show the case of $\lambda$-abstractions. 
There we have
\types{\Gamma}{\lambda[x:A]M}{\Pi[y:A]B}
which follows from 
\types{\Gamma}{A}{s_1}
and 
\types{\Gamma,z:A}{M[x:=z]}{B[y:=z]} (we omite some premises).
We must prove that there is some derivation of \Nf{(\lambda[x:A]M)}.
The proof is actually quite simple.
First, by renaming we have \nf{(M[x:=z])}. Second, by the IH we obtain a derivation for  \Nf{A} and another one for \Nf{(M[x:=z])}. Finally, by anti-renaming (\cref{lemma:antirenamingNf}) we obtain \Nf{M}, and by the abs. rule, \Nf{(\lambda[x:A]M)}.
\end{proof}

\subsection{The Proof of Consistency}

With the previous results at our disposal we can finally turn to the proof of consistency itself and show that, for some interesting class of impredicative and normalizing PTS to be defined, it follows that there is no term $M$ such that \types{}{M}{\Pi[x:s]x} can be derived for any $x$ and $s$. 
In order to do so, 
first, we will show a very rudimentary consistency result stating that there is no $M$ \textit{in normal form} such that \types{x:s}{M}{x}.
Leveraging on this, we will prove that neither can \types{}{M}{\Pi[x:s]z} be derived with $M$ in normal form too.
Finally, by (hypothetical) normalization and SR we will obtain the result desired for any $M$.

The next lemma will be required to deal with the case of variables in one of the steps. It states that for any typeable neutral-form $\lambda$-term, the variable at the head of the iteration has to be of functional type:
\begin{linklemma}{Section5/HeadFun.html}
\Verb|∀ {Γ x M N A} → Ne x M → Γ ⊢ₛ M · N ∶ A → ∃₄ λ y B C D → (x , B) ∈ Γ|\\
\Verb| × B ≃β Π[ y ∶ C ] D|
\label{lemma:headNeutralFun}
\end{linklemma}

\begin{proof}
By structural induction on the derivation of \Neu{x}{M}
and using \cref{lemma:inversionTypeApp,lemma:inversionTypeVar}.
\end{proof}

The next result is the one that prompted us to prove soundness of \Nfg/\Neg\ and as mentioned earlier:
\begin{linklemma}{Section5/ConvSort.html}
\label{lemma:normalTermConvertsToSort}
\Verb|∀ {A s} → nf A → A ≃β c s → A ≡ c s|
\end{linklemma}

\begin{proof}
By the second CR theorem we have there are some $\alpha$-convertible terms $B$ and $C$ such that \manystep{A}{B} and \manystep{s}{C}. By a local induction on the latter we can show that $C \equiv s$ must be the case, hence $B \equiv s$ as well. Now we analyze the derivation of the former sequence. On the one hand, if $A = B$ then we are done. 
On the other hand, if \manystep{A}{B} follows from \step{A}{D} and \manystep{D}{s} for some $D$, then we get a contradiction with the hypothesis that $A$ is in normal form, ex falso quodlibet.
\end{proof}

The direct alternative of using \Nfg\ instead of \nfg\ above implies to conduct the proof above by induction on the structure of \Nf{A}, simultaneously with a proof for a similar result about \Neu{}{A}. In addition to that, several results about the impossibility of existence of reductions beginning in normal-form terms (using \Nfg) would also be required. It does not seem worthwhile to consider this path. 

We continue with the following primitive consistency results:

\begin{linklemma}{Section5/PrimConsistency.html}
\label{lemma:primitiveConsistency}
\Verb|∀ {z s M} → Nf M → ¬([(z , c s)] ⊢ₛ M ∶ v z)|
\end{linklemma}

\begin{proof}
Suppose \types{z:s}{M}{z} does hold. Then we should arrive at some contradiction. We proceed by structural induction on the derivation of \Nf{M}. 
All cases are quite similar, each of them requiring an appropriate inversion lemma to obtain the aforementioned contradiction: if $M$ is a variable then we use \cref{lemma:inversionTypeVar} to find out \conv{z}{s'} for some $s'$, which by \cref{lemma:absurdEquationVarSort} is absurd; if it is an abstraction then we use \cref{lemma:inversionTypeAbs} and \cref{lemma:absurdEquationVarProd}; etc. 
Exceptionally, if $M$ is an application then we use \cref{lemma:headNeutralFun} instead to derive a contradictory equation.
\end{proof}

\begin{lemma}
\label{lemma:consistency}
\begin{enumerate}[ref=\thelemma.(\roman{enumi}),label=(\roman{enumi}),itemsep=0pt]
\myitem{Section5/ConsistencyNe.html} 
\Verb|∀ {M y A} → Ne y M → ¬([] ⊢ₛ M ∶ A)|

\myitem{Section5/ConsistencyNF.html}
\Verb|∀ {M x s} → Nf M → ¬([] ⊢ₛ M ∶ Π[ x ∶ c s ](v x))|
\label{lemma:consistencyNf}
\end{enumerate}
\end{lemma}

\begin{proof}
(i) follows by induction on the structure of \Neu{y}{M} and by using \cref{lemma:inversionnf}. As to (ii), we proceed by case analysis on the structure of \Nf{M}. We only show the case of $\lambda$-abstractions, which is the most complex.
There we have $M=\lambda[y:A]N$, and let us suppose there is a derivation of \types{}{\lambda[y:A]N}{\Pi[x:s]x}; we should get a contradiction.
By inversion on abstractions we can derive \types{z:A}{N[y:=z]}{B[y':=z]} for some $z$ with \conv{\Pi[x:s]x}{\Pi[y':A]B}. 
Next, by product injectivity we know \conv{s}{A} and \conv{z'}{B[y':=z']} for any $z'$, particularly \conv{z}{B[y':=z]}. 
By \cref{lemma:soundnessNf} we have \nf{A}, and so by \cref{lemma:normalTermConvertsToSort}, $A \equiv s$. 
Thus, so far, we have \types{z:s}{N[y:=z]}{B[y':=z]}. 
Now we would like to use \cref{lemma:primitiveConsistency}, so we must: (i) show that $N[y:=z]$ is also in normal form, and; (ii) convert the type $B[y':=z]$ to $z$. (i) follows by \cref{lemma:stabilityNfUnderRenaming}. 
As to (ii), by \cref{lemma:contextValidity} we have \ok{(z:s)}, hence by the variable rule, \types{z:s}{z}{s}, and thus, we can use the conversion rule to derive \types{z:s}{N[y:=z]}{z} and proceed as mentioned. 
\end{proof}

Without further ado, we have consistency. 
Let $(\mathcal{S},\mathcal{A},\mathcal{R})$ be any normalizing PTS.\footnote{\texttt{Normalizing} is a parameter or assumption to the current module.} Then:
\begin{theorem}[\sourcelink{Section5/Consistency.html}{Consistency}]
\label{theo:consistency}
\Verb|∀ {x s} → ¬(∃ λ M → [] ⊢ₛ M ∶ Π[ x ∶ c s ](v x))|
\end{theorem}

\begin{proof}
Suppose there is such term and derivation. By normalization we have there is some $N$ in normal form (\nf{N}) such that $M$ reduces to in a finite number of steps. 
By SR we have \types{}{N}{\Pi[x:s]x}. 
By \cref{lemma:completenessNf} we have \Nf{N}, and so by \cref{lemma:consistencyNf} we obtain a contradiction.
\end{proof}

\subsection{Validity of the Uninhabited Type and Impredicativity}

Is the type $\Pi[x:s]x$ valid? For if that were \textit{not} the case, then it would obviously follow that it cannot be assigned to any term and thus, the proof would not actually account for consistency but merely be a fallacy. 
The next lemma shows that $\Pi[x : s]x$ is certainly valid, and furthermore, it precisely defines those PTS to which the consistency proof can be applied:

\begin{linklemma}{Section5/ValidityFalsehood.html}
\Verb|∀ {x s t} → [] ⊢ₛ Π[ x ∶ c s ](v x) ∶ c t ↔ ∃ λ u → 𝒜 s u × ℛ u s t|
\end{linklemma}

\begin{proof}
The direct follows by using the inversion lemmas and \cref{lemma:normalTermConvertsToSort}, while the converse by the product and variable rules.
\end{proof}

The axioms and rules on the right-hand side of the equivalence above describe an impredicative principle.\footnote{This should come as no surprise since our proof is an adaptation of one for CC, which is impredicative.} 
Let us replace $s$ by $*$ and $u$ by $\square$ in order to establish some parallelism with the systems in the $\lambda$-cube. 
In these systems it is possible to use the product rule to derive the judgement below:
\[
\Pi[x:*]B:*
\]
If we associate $*$ to propositions, 
then the preceding judgment establishes that any quantification of propositions is also a proposition, or, read conversely, which is more suggestive, that the type of propositions ``contains'' (is inhabited by) any quantification of itself; i.e., $*$ is an impredicative definition.
Some well-known instances of the PTS, which also belong to the $\lambda$-cube, that satisfy $(*,\square)\in\mathcal{A}$ and $(\square,*,*)\in\mathcal{R}$ are:
$\lambda 2$ or System~F;
$\lambda$P$2$;
$\lambda \omega$;
and; $\lambda$C.
In these systems, $\Pi[x:{*}]x$ codes falsehood, which translates to $\forall X. X$ under the Curry-Howard interpretation. 

\section{Conclusions and Related Work}
\label{sec:conclusions}

All in all, the size of the code does not seem to have exploded by any means. 
The entire formalization is approximately 4300LoC, from which 3000LoC correspond to previous work on Stoughton's substitutions and PTS. 
To put in perspective, there are at least two works with which we can compare ours.
First, there is the formalization of the meta-theory of CC, PTS and CTS using dBI in Coq by Barras and Werner \cite{coqincoqSources,barrasSourcePTS,barras96a,barras96aShort}. 
Their corresponding formalization is roughly 2900LoC.
And second, there is the formalization by Aydemir~et~al. on the meta-theory of CC using locally-nameless syntax in Coq as well \cite{Aydemir2008sources,aydemir2008}. They stopped short of proving consistency, and the work took out about 4800LoC. 

A complementary method to the LoC count for evaluating the approach is to assess
the criteria from the POPLMark challenge~\cite{aydemir2005}: (1) \textit{overhead}, (2) \textit{transparency} and (3) \textit{cost of entry}. Certainly, (1) is medium at least or high if we agree with their judgment on the dBI technique (c.f.~p.~55).
As to (2), we take the view that our approach quite steals the show (at least in terms of the syntax used for presenting our results, not in the arguments used within the proofs) as it is the most faithful to informal presentations. Take for instance  \cref{lemma:weakening} and compare it with the following different notations:\footnote{The formalization of the weakening lemma is particularly transparent too when using locally-nameless syntax, however, other basic definitions using this technique such as $\beta$-reduction are not as due to some other notational clutter, e.g., the opening action, mentioned in the next paragraph.} 
\begin{table}[h]
\centering
\begin{tabular}{ccc}
\AxiomC{$\Gamma \vdash A : B$}
\AxiomC{$\Gamma \vdash C : s$}
\RightLabel{($x\not\in\Gamma$)}
\BinaryInfC{$\Gamma, x:C \vdash A : B$}
\DisplayProof    
&       
\AxiomC{$\Gamma \vdash A : B$}
\AxiomC{$\Gamma; C \vdash$}
\BinaryInfC{$\Gamma; C \vdash {\uparrow^1} A : {\uparrow^1} B$}
\DisplayProof                    
&       
\AxiomC{$\Gamma \vdash A : B$}
\AxiomC{$\vdash \Gamma \bullet C$}
\BinaryInfC{$\Gamma, C \vdash A[{\uparrow}\mathsf{id}] : B[{\uparrow}\mathsf{id}]$}
\DisplayProof   
\\[2ex]
Barendregt \cite{barendregt92} (informal, classical) 
& 
Barras \& Werner (dBI) 
& 
Abel et al. \cite{abel2018} (dBI) 
\end{tabular}
\end{table}
\\
We can see that our solution is almost identical to the classical one.
As to (3), in \cite{copes2018,urciuoli2020} it was reported that the formalizations present there were carried out successfully by Master students with little background on Agda, giving us evidence to support the claim that the cost of entry must be somewhat low.
All in all, we can say with enough conviction that our approach to type theory should guarantee success on future formalizations (the fundamental evaluation criterion from which the three already discussed are derived), specially if the lemmas about substitutions and $\alpha$-conversion on the syntax can be automated away as with the Autosubst framework \cite{stark2019}, an interesting challenge.

We believe our solution is also more transparent than those using locally-nameless syntax.
Take for instance the work by McKinna and Pollack, already described in the Introduction,
and consider the rule for abstractions of $\beta$-reduction in their style
\cite[p.~385]{mckinna99}: 
\[
(\xi)
\quad
q\not\in M,N
\,
\land
\,
\step{[q/x]M}{[q/y]N}
\Longrightarrow
\step{\lambda x M}{\lambda y N}
\]
$q$ is a free variable or \textit{parameter} while $x$ and $y$ are bound variables or just \textit{variables}. 
This rule is quite involved.
This is partly due to that fact that, since $M$ and $N$ are not well-formed terms by themselves, they have to be ``opened'' by renaming their respective variables to fresh parameters.
The other point is that a different variable must be used in each term to allow renaming one to the other in certain cases, 
for instance, in the case some redex has been contracted somewhere in the top of the derivation and as a result some occurrence of $x$ bound to an inner abstraction enters into the scope of the outter-most abstraction, e.g., 
while reducing the term
$\lambda x ((\lambda y \lambda x y) x)$. 
Actually, it turns out that this rule is the same to that for abstractions in $\alpha$-conversion.
The solution of using two sorts of names was suggested to them by Coquand in order to alleviate reasoning with $\alpha$-conversion.
However, we can see that the result is not exactly what we expected,
as the problem of renaming variables is not avoided, but rather postponed during substitution, 
and it ends up scattered throughout the development.

In order to ``complete'' the proof of \cref{theo:consistency} we would have to mechanize normalization. 
Normalization proofs for impredicative systems exist but require impredicative features in the meta-language (c.f. \cite{altenkirch1993,barras96a,barras96aShort,donelly2007}). 
Since Agda is \textit{not} impredicative, it does not seem possible to mechanize them.\footnote{It is worth mentioning that in \cite{chapman2019} the authors formalized (some properties of) an extension of System~F which is not normalizing at the \textit{term level} in Agda, however, they did mechanize normalization for the \textit{type level} which is isomorphic to STLC. Similarly, quite a few solutions \cite{poplmarksol} have been submitted in response to the POPLMark Challenge \cite{aydemir2005}, nevertheless, these sets of problems focuses on results such as  preservation and progress (SR) of System~F but not on termination or normalization.}
In fact, Girard has shown that normalization of System~F implies consistency of PA${}_2$ \cite{girard1989}. If such proof existed, then one would expect
Agda to be proof-theoretically stronger than PA${}_2$ so as not to contradict G\"odel's incompleteness theorem, which does not seem likely. 

\bibliographystyle{eptcs}
\bibliography{references}

\begin{thebibliography}{10}
\providecommand{\bibitemdeclare}[2]{}
\providecommand{\surnamestart}{}
\providecommand{\surnameend}{}
\providecommand{\urlprefix}{Available at }
\providecommand{\url}[1]{\texttt{#1}}
\providecommand{\href}[2]{\texttt{#2}}
\providecommand{\urlalt}[2]{\href{#1}{#2}}
\providecommand{\doi}[1]{doi:\urlalt{https://doi.org/#1}{#1}}
\providecommand{\eprint}[1]{arXiv:\urlalt{https://arxiv.org/abs/#1}{#1}}
\providecommand{\bibinfo}[2]{#2}

\bibitemdeclare{article}{abel2019}
\bibitem{abel2019}
\bibinfo{author}{Andreas \surnamestart Abel\surnameend},
  \bibinfo{author}{Guillaume \surnamestart Allais\surnameend},
  \bibinfo{author}{Aliya \surnamestart Hameer\surnameend},
  \bibinfo{author}{Brigitte \surnamestart Pientka\surnameend},
  \bibinfo{author}{Alberto \surnamestart Momigliano\surnameend},
  \bibinfo{author}{Steven \surnamestart Schäfer\surnameend} \&
  \bibinfo{author}{Kathrin \surnamestart Stark\surnameend}
  (\bibinfo{year}{2019}): \emph{\bibinfo{title}{{POPLMark reloaded: Mechanizing
  Proofs by Logical Relations}}}.
\newblock {\slshape \bibinfo{journal}{Journal of Functional Programming}}
  \bibinfo{volume}{29}, \doi{10.1017/S0956796819000170}.

\bibitemdeclare{article}{abel2018}
\bibitem{abel2018}
\bibinfo{author}{Andreas \surnamestart Abel\surnameend},
  \bibinfo{author}{Joakim \surnamestart {\"{O}}hman\surnameend} \&
  \bibinfo{author}{Andrea \surnamestart Vezzosi\surnameend}
  (\bibinfo{year}{2018}): \emph{\bibinfo{title}{{Decidability of Conversion for
  Type Theory in Type Theory}}}.
\newblock {\slshape \bibinfo{journal}{Proceedings of the ACM on Programing
  Languages}} \bibinfo{volume}{2}(\bibinfo{number}{POPL}), pp.
  \bibinfo{pages}{23:1--23:29}, \doi{10.1145/3158111}.

\bibitemdeclare{inproceedings}{altenkirch1993}
\bibitem{altenkirch1993}
\bibinfo{author}{Thorsten \surnamestart Altenkirch\surnameend}
  (\bibinfo{year}{1993}): \emph{\bibinfo{title}{A Formalization of the Strong
  Normalization Proof for System F in LEGO}}.
\newblock In \bibinfo{editor}{Marc \surnamestart Bezem\surnameend} \&
  \bibinfo{editor}{Jan~Friso \surnamestart Groote\surnameend}, editors:
  {\slshape \bibinfo{booktitle}{Proceedings of the International Conference on
  Typed Lambda Calculi and Applications (TLCA~'93)}},
  \bibinfo{publisher}{Springer Berlin Heidelberg}, p. \bibinfo{pages}{13–28},
  \doi{10.1007/bfb0037095}.

\bibitemdeclare{inproceedings}{aydemir2005}
\bibitem{aydemir2005}
\bibinfo{author}{Brian~E. \surnamestart Aydemir\surnameend},
  \bibinfo{author}{Aaron \surnamestart Bohannon\surnameend},
  \bibinfo{author}{Matthew \surnamestart Fairbairn\surnameend},
  \bibinfo{author}{J.~Nathan \surnamestart Foster\surnameend},
  \bibinfo{author}{Benjamin~C. \surnamestart Pierce\surnameend},
  \bibinfo{author}{Peter \surnamestart Sewell\surnameend},
  \bibinfo{author}{Dimitrios \surnamestart Vytiniotis\surnameend},
  \bibinfo{author}{Geoffrey \surnamestart Washburn\surnameend},
  \bibinfo{author}{Stephanie \surnamestart Weirich\surnameend} \&
  \bibinfo{author}{Steve \surnamestart Zdancewic\surnameend}
  (\bibinfo{year}{2005}): \emph{\bibinfo{title}{Mechanized Metatheory for the
  Masses: The POPLMark Challenge}}.
\newblock In \bibinfo{editor}{Joe \surnamestart Hurd\surnameend} \&
  \bibinfo{editor}{Tom \surnamestart Melham\surnameend}, editors: {\slshape
  \bibinfo{booktitle}{Proceedings of the 18th International Conference on
  Theorem Proving in Higher Order Logics (TPHOLs~'05)}},
  \bibinfo{publisher}{Springer Berlin Heidelberg}, p. \bibinfo{pages}{50–65},
  \doi{10.1007/11541868_4}.

\bibitemdeclare{misc}{Aydemir2008sources}
\bibitem{Aydemir2008sources}
\bibinfo{author}{Brian~E. \surnamestart Aydemir\surnameend},
  \bibinfo{author}{Arthur \surnamestart Chargu{\'{e}}raud\surnameend},
  \bibinfo{author}{Benjamin~C. \surnamestart Pierce\surnameend},
  \bibinfo{author}{Randy \surnamestart Pollack\surnameend} \&
  \bibinfo{author}{Stephanie \surnamestart Weirich\surnameend}:
  \emph{\bibinfo{title}{Sources of: Engineering Formal Metatheory}}.
\newblock
  \urlprefix\url{https://www.chargueraud.org/research/2007/binders/formal_binders.tar.gz}.

\bibitemdeclare{inproceedings}{aydemir2008}
\bibitem{aydemir2008}
\bibinfo{author}{Brian~E. \surnamestart Aydemir\surnameend},
  \bibinfo{author}{Arthur \surnamestart Chargu{\'{e}}raud\surnameend},
  \bibinfo{author}{Benjamin~C. \surnamestart Pierce\surnameend},
  \bibinfo{author}{Randy \surnamestart Pollack\surnameend} \&
  \bibinfo{author}{Stephanie \surnamestart Weirich\surnameend}
  (\bibinfo{year}{2008}): \emph{\bibinfo{title}{Engineering Formal
  Metatheory}}.
\newblock In \bibinfo{editor}{George~C. \surnamestart Necula\surnameend} \&
  \bibinfo{editor}{Philip \surnamestart Wadler\surnameend}, editors: {\slshape
  \bibinfo{booktitle}{Proceedings of the 35th {ACM} {SIGPLAN-SIGACT} Symposium
  on Principles of Programming Languages (POPL~'08)}},
  \bibinfo{publisher}{{ACM}}, pp. \bibinfo{pages}{3--15},
  \doi{10.1145/1328438.1328443}.

\bibitemdeclare{book}{barendregt84}
\bibitem{barendregt84}
\bibinfo{author}{Hendrik~P. \surnamestart Barendregt\surnameend}
  (\bibinfo{year}{1984}): \emph{\bibinfo{title}{The Lambda Calculus: Its Syntax
  and Semantics}}, \bibinfo{edition}{revised} edition.
\newblock {\slshape \bibinfo{series}{Studies in Logic and the Foundations of
  Mathematics}} \bibinfo{volume}{103}, \bibinfo{publisher}{North-Holland}.

\bibitemdeclare{article}{barendregt1991}
\bibitem{barendregt1991}
\bibinfo{author}{Hendrik~P. \surnamestart Barendregt\surnameend}
  (\bibinfo{year}{1991}): \emph{\bibinfo{title}{Introduction to Generalized
  Type Systems}}.
\newblock {\slshape \bibinfo{journal}{Journal of Functional Programming}}
  \bibinfo{volume}{1}(\bibinfo{number}{2}), p. \bibinfo{pages}{125–154},
  \doi{10.1017/S0956796800020025}.

\bibitemdeclare{incollection}{barendregt92}
\bibitem{barendregt92}
\bibinfo{author}{Hendrik~P. \surnamestart Barendregt\surnameend}
  (\bibinfo{year}{1992}): \emph{\bibinfo{title}{Lambda Calculi with Types}}.
\newblock In \bibinfo{editor}{S.~\surnamestart Abramsky\surnameend},
  \bibinfo{editor}{Dov~M. \surnamestart Gabbay\surnameend} \&
  \bibinfo{editor}{T.~S.~E. \surnamestart Maibaum\surnameend}, editors:
  {\slshape \bibinfo{booktitle}{Handbook of Logic in Computer Science}},
  \bibinfo{volume}{2}, \bibinfo{publisher}{Oxford University Press}, p.
  \bibinfo{pages}{117–309}, \doi{10.1093/oso/9780198537618.003.0002}.

\bibitemdeclare{misc}{coqincoqSources}
\bibitem{coqincoqSources}
\bibinfo{author}{Bruno \surnamestart Barras\surnameend}:
  \emph{\bibinfo{title}{{Sources of the CC Formalization}}}.
\newblock \urlprefix\url{https://github.com/rocq-archive/coq-in-coq}.

\bibitemdeclare{misc}{barrasSourcePTS}
\bibitem{barrasSourcePTS}
\bibinfo{author}{Bruno \surnamestart Barras\surnameend}:
  \emph{\bibinfo{title}{{Sources of the PTS Formalization}}}.
\newblock \urlprefix\url{https://github.com/rocq-archive/pts}.

\bibitemdeclare{techreport}{barras96a}
\bibitem{barras96a}
\bibinfo{author}{Bruno \surnamestart Barras\surnameend} (\bibinfo{year}{1996}):
  \emph{\bibinfo{title}{{Coq en Coq}}}.
\newblock \bibinfo{type}{Rapport de Recherche} \bibinfo{number}{3026},
  \bibinfo{institution}{INRIA}.

\bibitemdeclare{inproceedings}{barras96b}
\bibitem{barras96b}
\bibinfo{author}{Bruno \surnamestart Barras\surnameend} (\bibinfo{year}{1996}):
  \emph{\bibinfo{title}{Verification of the Interface of a Small Proof System
  in Coq}}.
\newblock In \bibinfo{editor}{Eduardo \surnamestart Giménez\surnameend} \&
  \bibinfo{editor}{Christine \surnamestart Paulin-Mohring\surnameend}, editors:
  {\slshape \bibinfo{booktitle}{Proceedings of the 1996 International Workshop
  on Types for Proofs and Programs (TYPES~'96)}}, {\slshape
  \bibinfo{series}{LNCS}} \bibinfo{volume}{1512}, \bibinfo{publisher}{Springer
  Berlin Heidelberg}, pp. \bibinfo{pages}{28--45}, \doi{10.1007/bfb0097785}.

\bibitemdeclare{unpublished}{barras96aShort}
\bibitem{barras96aShort}
\bibinfo{author}{Bruno \surnamestart Barras\surnameend} \&
  \bibinfo{author}{Benjamin \surnamestart Werner\surnameend}:
  \emph{\bibinfo{title}{{Coq in Coq}}}.
\newblock
  \urlprefix\url{https://www.lix.polytechnique.fr/page/index.php?username=barras&path=publi/coqincoq.pdf}.
\newblock \bibinfo{note}{Manuscript}.

\bibitemdeclare{inproceedings}{chapman2019}
\bibitem{chapman2019}
\bibinfo{author}{James \surnamestart Chapman\surnameend},
  \bibinfo{author}{Roman \surnamestart Kireev\surnameend},
  \bibinfo{author}{Chad \surnamestart Nester\surnameend} \&
  \bibinfo{author}{Philip \surnamestart Wadler\surnameend}
  (\bibinfo{year}{2019}): \emph{\bibinfo{title}{System F in Agda, for Fun and
  Profit}}.
\newblock In \bibinfo{editor}{Graham \surnamestart Hutton\surnameend}, editor:
  {\slshape \bibinfo{booktitle}{Proceedings of the 13th International
  Conference on Mathematics of Program Construction (MPC~'19)}}, {\slshape
  \bibinfo{series}{LNCS}} \bibinfo{volume}{11825}, \bibinfo{publisher}{Springer
  Berlin Heidelberg}, p. \bibinfo{pages}{255–297},
  \doi{10.1007/978-3-030-33636-3_10}.

\bibitemdeclare{article}{chargueraud2012}
\bibitem{chargueraud2012}
\bibinfo{author}{Arthur \surnamestart Chargu{\'e}raud\surnameend}
  (\bibinfo{year}{2012}): \emph{\bibinfo{title}{The Locally Nameless
  Representation}}.
\newblock {\slshape \bibinfo{journal}{Journal of Automated Reasoning}}
  \bibinfo{volume}{49}(\bibinfo{number}{3}), pp. \bibinfo{pages}{363--408},
  \doi{10.1007/s10817-011-9225-2}.

\bibitemdeclare{article}{copello2017}
\bibitem{copello2017}
\bibinfo{author}{Ernesto \surnamestart Copello\surnameend},
  \bibinfo{author}{Nora \surnamestart Szasz\surnameend} \&
  \bibinfo{author}{\surnamestart Álvaro Tasistro\surnameend}
  (\bibinfo{year}{2017}): \emph{\bibinfo{title}{Formal metatheory of the Lambda
  calculus using Stoughton's substitution}}.
\newblock {\slshape \bibinfo{journal}{Theoretical Computer Science}}
  \bibinfo{volume}{685}, pp. \bibinfo{pages}{65--82},
  \doi{10.1016/j.tcs.2016.08.025}.
\newblock \bibinfo{note}{Logical and Semantic Frameworks with Applications}.

\bibitemdeclare{inproceedings}{copes2018}
\bibitem{copes2018}
\bibinfo{author}{Mart{\'{\i}}n \surnamestart Copes\surnameend},
  \bibinfo{author}{Nora \surnamestart Szasz\surnameend} \&
  \bibinfo{author}{{\'{A}}lvaro \surnamestart Tasistro\surnameend}
  (\bibinfo{year}{2018}): \emph{\bibinfo{title}{Formalization in Constructive
  Type Theory of the Standardization Theorem for the Lambda Calculus using
  Multiple Substitution}}.
\newblock In \bibinfo{editor}{Fr{\'{e}}d{\'{e}}ric \surnamestart
  Blanqui\surnameend} \& \bibinfo{editor}{Giselle \surnamestart
  Reis\surnameend}, editors: {\slshape \bibinfo{booktitle}{Proceedings of the
  13th International Workshop on Logical Frameworks and Meta-Languages: Theory
  and Practice (LFMTP~'18)}}, {\slshape \bibinfo{series}{{EPTCS}}}
  \bibinfo{volume}{274}, pp. \bibinfo{pages}{27--41},
  \doi{10.4204/EPTCS.274.3}.

\bibitemdeclare{techreport}{coquand1990}
\bibitem{coquand1990}
\bibinfo{author}{Thierry \surnamestart Coquand\surnameend}
  (\bibinfo{year}{1990}): \emph{\bibinfo{title}{A Proof Of Strong Normalization
  For The Theory Of Constructions Using A Kripe-Like Interpretation}}.
\newblock \bibinfo{type}{Technical Report} \bibinfo{number}{MS-CIS-90-44},
  \bibinfo{institution}{University of Pennsylvania}.
\newblock
  \urlprefix\url{https://repository.upenn.edu/handle/20.500.14332/7509}.

\bibitemdeclare{article}{coquand1988}
\bibitem{coquand1988}
\bibinfo{author}{Thierry \surnamestart Coquand\surnameend} \&
  \bibinfo{author}{Gérard \surnamestart Huet\surnameend}
  (\bibinfo{year}{1988}): \emph{\bibinfo{title}{The Calculus of
  Constructions}}.
\newblock {\slshape \bibinfo{journal}{Information and Computation}}
  \bibinfo{volume}{76}(\bibinfo{number}{2}), pp. \bibinfo{pages}{95--120},
  \doi{10.1016/0890-5401(88)90005-3}.
\newblock
  \urlprefix\url{https://www.sciencedirect.com/science/article/pii/0890540188900053}.

\bibitemdeclare{book}{curry1958}
\bibitem{curry1958}
\bibinfo{author}{Haskell~B. \surnamestart Curry\surnameend} \&
  \bibinfo{author}{Robert \surnamestart Feys\surnameend}
  (\bibinfo{year}{1958}): \emph{\bibinfo{title}{Combinatory Logic}}.
\newblock \bibinfo{volume}{1}, \bibinfo{publisher}{North-Holland Publishing
  Company}.

\bibitemdeclare{article}{debruijn1972}
\bibitem{debruijn1972}
\bibinfo{author}{Nicolaas~G. \surnamestart {de Bruijn}\surnameend}
  (\bibinfo{year}{1972}): \emph{\bibinfo{title}{Lambda calculus notation with
  nameless dummies, a tool for automatic formula manipulation, with application
  to the Church-Rosser theorem}}.
\newblock {\slshape \bibinfo{journal}{Indagationes Mathematicae (Proceedings)}}
  \bibinfo{volume}{75}(\bibinfo{number}{5}), pp. \bibinfo{pages}{381--392},
  \doi{10.1016/1385-7258(72)90034-0}.

\bibitemdeclare{inproceedings}{despeyroux1995}
\bibitem{despeyroux1995}
\bibinfo{author}{Jo{\"e}lle \surnamestart Despeyroux\surnameend},
  \bibinfo{author}{Amy \surnamestart Felty\surnameend} \&
  \bibinfo{author}{Andr{\'e} \surnamestart Hirschowitz\surnameend}
  (\bibinfo{year}{1995}): \emph{\bibinfo{title}{Higher-order Abstract Syntax in
  Coq}}.
\newblock In \bibinfo{editor}{Mariangiola \surnamestart
  Dezani-Ciancaglini\surnameend} \& \bibinfo{editor}{Gordon \surnamestart
  Plotkin\surnameend}, editors: {\slshape \bibinfo{booktitle}{Proceedings of
  the Second International Conference on Typed Lambda Calculi and Applications
  (TLCA~'95)}}, {\slshape \bibinfo{series}{LNCS}} \bibinfo{volume}{902},
  \bibinfo{publisher}{Springer Berlin Heidelberg}, pp.
  \bibinfo{pages}{124--138}, \doi{10.1007/bfb0014049}.

\bibitemdeclare{article}{donelly2007}
\bibitem{donelly2007}
\bibinfo{author}{Kevin \surnamestart Donnelly\surnameend} \&
  \bibinfo{author}{Hongwei \surnamestart Xi\surnameend} (\bibinfo{year}{2007}):
  \emph{\bibinfo{title}{A Formalization of Strong Normalization for
  Simply-Typed Lambda-Calculus and System F}}.
\newblock {\slshape \bibinfo{journal}{Electronic Notes in Theoretical Computer
  Science}} \bibinfo{volume}{174}(\bibinfo{number}{5}), pp.
  \bibinfo{pages}{109--125}, \doi{10.1016/j.entcs.2007.01.021}.
\newblock \bibinfo{note}{Proceedings of the First International Workshop on
  Logical Frameworks and Meta-Languages: Theory and Practice (LFMTP~'06)}.

\bibitemdeclare{book}{girard1989}
\bibitem{girard1989}
\bibinfo{author}{Jean-Yves \surnamestart Girard\surnameend},
  \bibinfo{author}{Paul \surnamestart Taylor\surnameend} \&
  \bibinfo{author}{Yves \surnamestart Lafont\surnameend}
  (\bibinfo{year}{1989}): \emph{\bibinfo{title}{Proofs and Types}}.
\newblock \bibinfo{publisher}{Cambridge University Press}.

\bibitemdeclare{article}{harper93}
\bibitem{harper93}
\bibinfo{author}{Robert \surnamestart Harper\surnameend},
  \bibinfo{author}{Furio \surnamestart Honsell\surnameend} \&
  \bibinfo{author}{Gordon \surnamestart Plotkin\surnameend}
  (\bibinfo{year}{1993}): \emph{\bibinfo{title}{A Framework for Defining
  Logics}}.
\newblock {\slshape \bibinfo{journal}{Journal of the ACM}}
  \bibinfo{volume}{40}(\bibinfo{number}{1}), p. \bibinfo{pages}{143–184},
  \doi{10.1145/138027.138060}.

\bibitemdeclare{article}{harper2005}
\bibitem{harper2005}
\bibinfo{author}{Robert \surnamestart Harper\surnameend} \&
  \bibinfo{author}{Frank \surnamestart Pfenning\surnameend}
  (\bibinfo{year}{2005}): \emph{\bibinfo{title}{On Equivalence and Canonical
  Forms in the LF Type Theory}}.
\newblock {\slshape \bibinfo{journal}{ACM Transactions on Computational Logic}}
  \bibinfo{volume}{6}(\bibinfo{number}{1}), p. \bibinfo{pages}{61–101},
  \doi{10.1145/1042038.1042041}.

\bibitemdeclare{book}{hindley1997}
\bibitem{hindley1997}
\bibinfo{author}{J.~Roger \surnamestart Hindley\surnameend}
  (\bibinfo{year}{1997}): \emph{\bibinfo{title}{Basic Simple Type Theory}}.
\newblock {\slshape \bibinfo{series}{Cambridge Tracts in Theoretical Computer
  Science}}~\bibinfo{volume}{42}, \bibinfo{publisher}{Cambridge University
  Press}, \doi{10.1017/cbo9780511608865}.

\bibitemdeclare{techreport}{kashima2000}
\bibitem{kashima2000}
\bibinfo{author}{Ryo \surnamestart Kashima\surnameend} (\bibinfo{year}{2000}):
  \emph{\bibinfo{title}{{A Proof of the Standardization Theorem in
  Lambda-Calculus}}}.
\newblock \bibinfo{type}{Technical Report} \bibinfo{number}{RRMCS C-145},
  \bibinfo{institution}{Tokyo Institute of Technology}.

\bibitemdeclare{inproceedings}{mckinna93}
\bibitem{mckinna93}
\bibinfo{author}{James \surnamestart McKinna\surnameend} \&
  \bibinfo{author}{Robert \surnamestart Pollack\surnameend}
  (\bibinfo{year}{1993}): \emph{\bibinfo{title}{Pure Type Systems Formalized}}.
\newblock In \bibinfo{editor}{Marc \surnamestart Bezem\surnameend} \&
  \bibinfo{editor}{Jan~Friso \surnamestart Groote\surnameend}, editors:
  {\slshape \bibinfo{booktitle}{Proceedings of the International Conference on
  Typed Lambda Calculi and Applications (TLCA~'93)}},
  \bibinfo{publisher}{Springer Berlin Heidelberg}, pp.
  \bibinfo{pages}{289--305}, \doi{10.1007/BFb0037113}.

\bibitemdeclare{article}{mckinna99}
\bibitem{mckinna99}
\bibinfo{author}{James \surnamestart McKinna\surnameend} \&
  \bibinfo{author}{Robert \surnamestart Pollack\surnameend}
  (\bibinfo{year}{1999}): \emph{\bibinfo{title}{Some Lambda Calculus and Type
  Theory Formalized}}.
\newblock {\slshape \bibinfo{journal}{Journal of Automated Reasoning}}
  \bibinfo{volume}{23}(\bibinfo{number}{3}), pp. \bibinfo{pages}{373--409},
  \doi{10.1023/A:1006294005493}.

\bibitemdeclare{book}{nipkow2002}
\bibitem{nipkow2002}
\bibinfo{author}{Tobias \surnamestart Nipkow\surnameend},
  \bibinfo{author}{Markus \surnamestart Wenzel\surnameend} \&
  \bibinfo{author}{Lawrence~C. \surnamestart Paulson\surnameend}
  (\bibinfo{year}{2002}): \emph{\bibinfo{title}{Isabelle/HOL: A Proof Assistant
  for Higher-order Logic}}.
\newblock {\slshape \bibinfo{series}{LNCS}} \bibinfo{volume}{2283},
  \bibinfo{publisher}{Springer Berlin Heidelberg}, \doi{10.1007/3-540-45949-9}.

\bibitemdeclare{article}{pfenning1988}
\bibitem{pfenning1988}
\bibinfo{author}{Frank \surnamestart Pfenning\surnameend} \&
  \bibinfo{author}{Conal \surnamestart Elliott\surnameend}
  (\bibinfo{year}{1988}): \emph{\bibinfo{title}{Higher-order Abstract Syntax}}.
\newblock {\slshape \bibinfo{journal}{SIGPLAN Notices}}
  \bibinfo{volume}{23}(\bibinfo{number}{7}), p. \bibinfo{pages}{199–208},
  \doi{10.1145/960116.54010}.

\bibitemdeclare{phdthesis}{pollack94phd}
\bibitem{pollack94phd}
\bibinfo{author}{Robert \surnamestart Pollack\surnameend}
  (\bibinfo{year}{1994}): \emph{\bibinfo{title}{The Theory of LEGO}}.
\newblock Ph.D. thesis, \bibinfo{school}{University of Edinburgh}.
\newblock \urlprefix\url{https://era.ed.ac.uk/handle/1842/504}.

\bibitemdeclare{book}{sorensen2006}
\bibitem{sorensen2006}
\bibinfo{author}{Morten~H. \surnamestart S\o{}rensen\surnameend} \&
  \bibinfo{author}{Pawel \surnamestart Urzyczyn\surnameend}
  (\bibinfo{year}{2006}): \emph{\bibinfo{title}{Lectures on the Curry-Howard
  Isomorphism}}.
\newblock {\slshape \bibinfo{series}{Studies in Logic and the Foundations of
  Mathematics}} \bibinfo{volume}{149}, \bibinfo{publisher}{Elsevier}.

\bibitemdeclare{article}{sozeau2025}
\bibitem{sozeau2025}
\bibinfo{author}{Matthieu \surnamestart Sozeau\surnameend},
  \bibinfo{author}{Yannick \surnamestart Forster\surnameend},
  \bibinfo{author}{Meven \surnamestart Lennon-Bertrand\surnameend},
  \bibinfo{author}{Jakob \surnamestart Nielsen\surnameend},
  \bibinfo{author}{Nicolas \surnamestart Tabareau\surnameend} \&
  \bibinfo{author}{Th\'{e}o \surnamestart Winterhalter\surnameend}
  (\bibinfo{year}{2025}): \emph{\bibinfo{title}{Correct and Complete Type
  Checking and Certified Erasure for Coq, in Coq}}.
\newblock {\slshape \bibinfo{journal}{Journal of the ACM}}
  \bibinfo{volume}{72}(\bibinfo{number}{1}), \doi{10.1145/3706056}.

\bibitemdeclare{inproceedings}{stark2019}
\bibitem{stark2019}
\bibinfo{author}{Kathrin \surnamestart Stark\surnameend},
  \bibinfo{author}{Steven \surnamestart Sch\"{a}fer\surnameend} \&
  \bibinfo{author}{Jonas \surnamestart Kaiser\surnameend}
  (\bibinfo{year}{2019}): \emph{\bibinfo{title}{Autosubst 2: reasoning with
  multi-sorted de Bruijn terms and vector substitutions}}.
\newblock In: {\slshape \bibinfo{booktitle}{Proceedings of the 8th ACM SIGPLAN
  International Conference on Certified Programs and Proofs (CPP~'19)}},
  \bibinfo{publisher}{ACM}, p. \bibinfo{pages}{166–180},
  \doi{10.1145/3293880.3294101}.

\bibitemdeclare{article}{stoughton1988}
\bibitem{stoughton1988}
\bibinfo{author}{Allen \surnamestart Stoughton\surnameend}
  (\bibinfo{year}{1988}): \emph{\bibinfo{title}{Substitution Revisited}}.
\newblock {\slshape \bibinfo{journal}{Theoretical Computer Science}}
  \bibinfo{volume}{59}(\bibinfo{number}{3}), pp. \bibinfo{pages}{317--325},
  \doi{10.1016/0304-3975(88)90149-1}.

\bibitemdeclare{article}{takahashi1995}
\bibitem{takahashi1995}
\bibinfo{author}{Masako \surnamestart Takahashi\surnameend}
  (\bibinfo{year}{1995}): \emph{\bibinfo{title}{Parallel Reductions in
  λ-Calculus}}.
\newblock {\slshape \bibinfo{journal}{Information and Computation}}
  \bibinfo{volume}{118}(\bibinfo{number}{1}), pp. \bibinfo{pages}{120--127},
  \doi{10.1006/inco.1995.1057}.

\bibitemdeclare{inproceedings}{tasistro2015}
\bibitem{tasistro2015}
\bibinfo{author}{\'Alvaro \surnamestart Tasistro\surnameend},
  \bibinfo{author}{Ernesto \surnamestart Copello\surnameend} \&
  \bibinfo{author}{Nora \surnamestart Szasz\surnameend} (\bibinfo{year}{2015}):
  \emph{\bibinfo{title}{Formalisation in Constructive Type Theory of
  Stoughton's Substitution for the Lambda Calculus}}.
\newblock In \bibinfo{editor}{Mauricio \surnamestart Ayala-Rinc\'on\surnameend}
  \& \bibinfo{editor}{Ian \surnamestart Mackie\surnameend}, editors: {\slshape
  \bibinfo{booktitle}{Proceedings of the Ninth Workshop on Logical and Semantic
  Frameworks, with Applications (LSFA~'14)}}, {\slshape
  \bibinfo{series}{ENTCS}} \bibinfo{volume}{312},
  \bibinfo{publisher}{Elsevier}, pp. \bibinfo{pages}{215--230},
  \doi{10.1016/j.entcs.2015.04.013}.

\bibitemdeclare{misc}{poplmarksol}
\bibitem{poplmarksol}
\bibinfo{author}{The~POPLMark \surnamestart Team\surnameend}:
  \urlprefix\url{https://www.seas.upenn.edu/~plclub/poplmark/}.

\bibitemdeclare{misc}{agda2622}
\bibitem{agda2622}
\bibinfo{author}{\surnamestart {The Agda Team}\surnameend}:
  \emph{\bibinfo{title}{Agda v2.6.2.2}}.
\newblock \urlprefix\url{https://agda.readthedocs.io/en/v2.6.2.2/}.

\bibitemdeclare{misc}{agdastdlib171}
\bibitem{agdastdlib171}
\bibinfo{author}{\surnamestart {The Agda's Standard Library Team}\surnameend}:
  \emph{\bibinfo{title}{The Standard Library v1.7.1}}.
\newblock
  \urlprefix\url{https://github.com/agda/agda-stdlib/releases/tag/v1.7.1}.

\bibitemdeclare{article}{urban2011}
\bibitem{urban2011}
\bibinfo{author}{Christian \surnamestart Urban\surnameend},
  \bibinfo{author}{James \surnamestart Cheney\surnameend} \&
  \bibinfo{author}{Stefan \surnamestart Berghofer\surnameend}
  (\bibinfo{year}{2011}): \emph{\bibinfo{title}{Mechanizing the metatheory of
  LF}}.
\newblock {\slshape \bibinfo{journal}{ACM Transactions on Computational Logic}}
  \bibinfo{volume}{12}(\bibinfo{number}{2}), \doi{10.1145/1877714.1877721}.

\bibitemdeclare{inproceedings}{urciuoli2023}
\bibitem{urciuoli2023}
\bibinfo{author}{Sebasti\'an \surnamestart Urciuoli\surnameend}
  (\bibinfo{year}{2023}): \emph{\bibinfo{title}{A Formal Proof of the Strong
  Normalization Theorem for System T in Agda}}.
\newblock In \bibinfo{editor}{Daniele \surnamestart Nantes-Sobrinho\surnameend}
  \& \bibinfo{editor}{Pascal \surnamestart Fontaine\surnameend}, editors:
  {\slshape \bibinfo{booktitle}{Proceedings of the 17th International Workshop
  on Logical and Semantic Frameworks with Applications (LSFA~'22)}}, {\slshape
  \bibinfo{series}{EPTCS}} \bibinfo{volume}{376}, \bibinfo{publisher}{Open
  Publishing Association}, p. \bibinfo{pages}{81–99},
  \doi{10.4204/eptcs.376.8}.

\bibitemdeclare{misc}{sources}
\bibitem{sources}
\bibinfo{author}{Sebasti\'an \surnamestart Urciuoli\surnameend}
  (\bibinfo{year}{2026}): \emph{\bibinfo{title}{Sources}},
  \doi{10.5281/zenodo.20586635}.

\bibitemdeclare{misc}{html}
\bibitem{html}
\bibinfo{author}{Sebasti\'an \surnamestart Urciuoli\surnameend}
  (\bibinfo{year}{2026}): \emph{\bibinfo{title}{Sources (HTML)}}.
\newblock \urlprefix\url{https://surciuoli.github.io/pts-consistency-impred}.

\bibitemdeclare{inproceedings}{urciuoli2020}
\bibitem{urciuoli2020}
\bibinfo{author}{Sebasti\'an \surnamestart Urciuoli\surnameend},
  \bibinfo{author}{\'Alvaro \surnamestart Tasistro\surnameend} \&
  \bibinfo{author}{Nora \surnamestart Szasz\surnameend} (\bibinfo{year}{2020}):
  \emph{\bibinfo{title}{Strong Normalization for the Simply-Typed Lambda
  Calculus in Constructive Type Theory Using Agda}}.
\newblock In \bibinfo{editor}{Cláudia \surnamestart Nalon\surnameend} \&
  \bibinfo{editor}{Giselle \surnamestart Reis\surnameend}, editors: {\slshape
  \bibinfo{booktitle}{Proceedings of the 15th International Workshop on Logical
  and Semantic Frameworks, with Applications (LSFA~'20)}}, {\slshape
  \bibinfo{series}{ENTCS}} \bibinfo{volume}{351},
  \bibinfo{publisher}{Elsevier}, pp. \bibinfo{pages}{187--203},
  \doi{10.1016/j.entcs.2020.08.010}.

\bibitemdeclare{inproceedings}{urciuoli2025}
\bibitem{urciuoli2025}
\bibinfo{author}{Sebastián \surnamestart Urciuoli\surnameend}
  (\bibinfo{year}{2025}): \emph{\bibinfo{title}{On the Formal Metatheory of the
  Pure Type Systems using One-sorted Variable Names and Multiple
  Substitutions}}.
\newblock In \bibinfo{editor}{Kaustuv \surnamestart Chaudhuri\surnameend} \&
  \bibinfo{editor}{Daniele \surnamestart Nantes-Sobrinho\surnameend}, editors:
  {\slshape \bibinfo{booktitle}{Proceedings of the Twentieth International
  Workshop on Logical Frameworks and Meta-Languages: Theory and Practice
  (LFMTP~'25)}}, {\slshape \bibinfo{series}{EPTCS}} \bibinfo{volume}{431},
  \bibinfo{publisher}{Open Publishing Association}, p.
  \bibinfo{pages}{17–33}, \doi{10.4204/eptcs.431.2}.

\end{thebibliography}

\end{document}